%% file: main.tex
\documentclass{article}
\usepackage{iclr2027_conference,times}

\input{math_commands.tex}

\renewcommand{\eqref}[1]{\textup{(\ref{#1})}}
\usepackage{longtable,array}
\usepackage{hyperref}
\usepackage{url}
\usepackage{booktabs}
\usepackage{multirow}
\usepackage{graphicx}
\usepackage{microtype}
\usepackage{amsthm}
\usepackage{amssymb}
\usepackage{graphicx}
\usepackage{wrapfig}
\usepackage{tabularx}
\usepackage{algorithm}
\usepackage{algpseudocode}
\usepackage{xcolor}
\DeclareUnicodeCharacter{266B}{}

\newtheorem{proposition}{Proposition}
\AtBeginDocument{\setlength{\abovedisplayskip}{5pt}\setlength{\belowdisplayskip}{5pt}\setlength{\abovedisplayshortskip}{3pt}\setlength{\belowdisplayshortskip}{3pt}}

\makeatletter
\renewenvironment{proof}[1][\proofname]{\par
  \pushQED{\qed}%
  \normalfont \topsep0pt \partopsep0pt \trivlist
  \item[\hskip\labelsep\itshape #1\@addpunct{.}]\ignorespaces
}{\popQED\endtrivlist\@endpefalse}
\makeatother

\title{\hyphenpenalty=10000 CollabFlow: Recursive Self-Improvement\\of Agent Collaboration}

\author{\parbox[t]{\dimexpr\textwidth-2\tabcolsep\relax}{\centering
  Xiao Huang$^{1,*}$,\; Mingda Zhang$^{2,*}$,\; Junming Zhang$^{1}$,\; Qiang Huang$^{3}$,\; Hanwen Zhang$^{2}$, \\
  Yue Dai$^{2}$,\; Zijia Wang$^{4}$,\; Xiaoying Tang$^{2}$ \\[3pt]
  {\normalfont $^{1}$Shanghai Jiao Tong University \quad
  $^{2}$The Chinese University of Hong Kong, Shenzhen \\
  $^{3}$Fudan University \quad
  $^{4}$University of Oxford \quad
  $^{*}$Equal contribution \\[3pt]
\texttt{mingdazhang@ieee.org}}}}
\iclrfinalcopy


\pdfpageattr{/Group << /S /Transparency /I true /CS /DeviceRGB >>}
\begin{document}
\vspace*{-25pt}
\maketitle\vspace{-20pt}
\lhead{}

\begin{abstract}
Recursive self-improvement (RSI) lets a system improve from its own outcomes; in LLM-based multi-agent systems, Agents refine one another within a task, and outcomes improve how they collaborate across tasks.
However, existing multi-agent collaboration leaves this loop open: collaboration is pre-defined at the operator level, topology-only learning keeps verbatim exchange that propagates errors, and reward maximization on a system's own outcomes concentrates on a few teams.
To address these challenges, we propose CollabFlow, an RSI system of Learned Agent Collaboration: a trainable Collab-Director constructs teams of complete Agents, a frozen executor runs them, and each round's outcomes retrain the director.
Within each round, the edges of a collaboration graph carry protocols of Evidence-Conditioned Communication: a receiver adopts a differing answer only when the sender's evidence is stronger by a margin, so the director learns who communicates and how.
Across rounds, we further propose Collaborative Trajectory Balance (CTB), a flow-based objective that credits each team once across its construction orders and targets a reward-proportional distribution over teams, so several good teams stay in play.
We also bound how far this self-generated target moves between rounds, which shrinks as records accumulate.
On twelve datasets, CollabFlow outperforms all baselines and keeps improving across rounds.
Code is available at \url{https://anonymous.4open.science/r/CollabFlow-631E}.
\end{abstract}

\input{sections/01_introduction}
\input{sections/02_related_work}
\input{sections/03_method}
\input{sections/04_methodology}
\input{sections/04_experiments}
\input{sections/05_conclusion}

\bibliography{references}
\bibliographystyle{iclr2027_conference}

\appendix
\input{sections/appendix}

\end{document}

%% file: math_commands.tex
\usepackage{amsmath,amsfonts,bm}

\def\eqref#1{equation~\ref{#1}}

\def\1{\bm{1}}

\DeclareMathAlphabet{\mathsfit}{\encodingdefault}{\sfdefault}{m}{sl}
\SetMathAlphabet{\mathsfit}{bold}{\encodingdefault}{\sfdefault}{bx}{n}

%% file: sections/01_introduction.tex
\section{Introduction}
\label{sec:introduction}
\suppressfloats[t]
\begin{wrapfigure}[18]{r}{0.50\textwidth}
\centering
\vspace{-5pt}
\includegraphics[width=\linewidth]{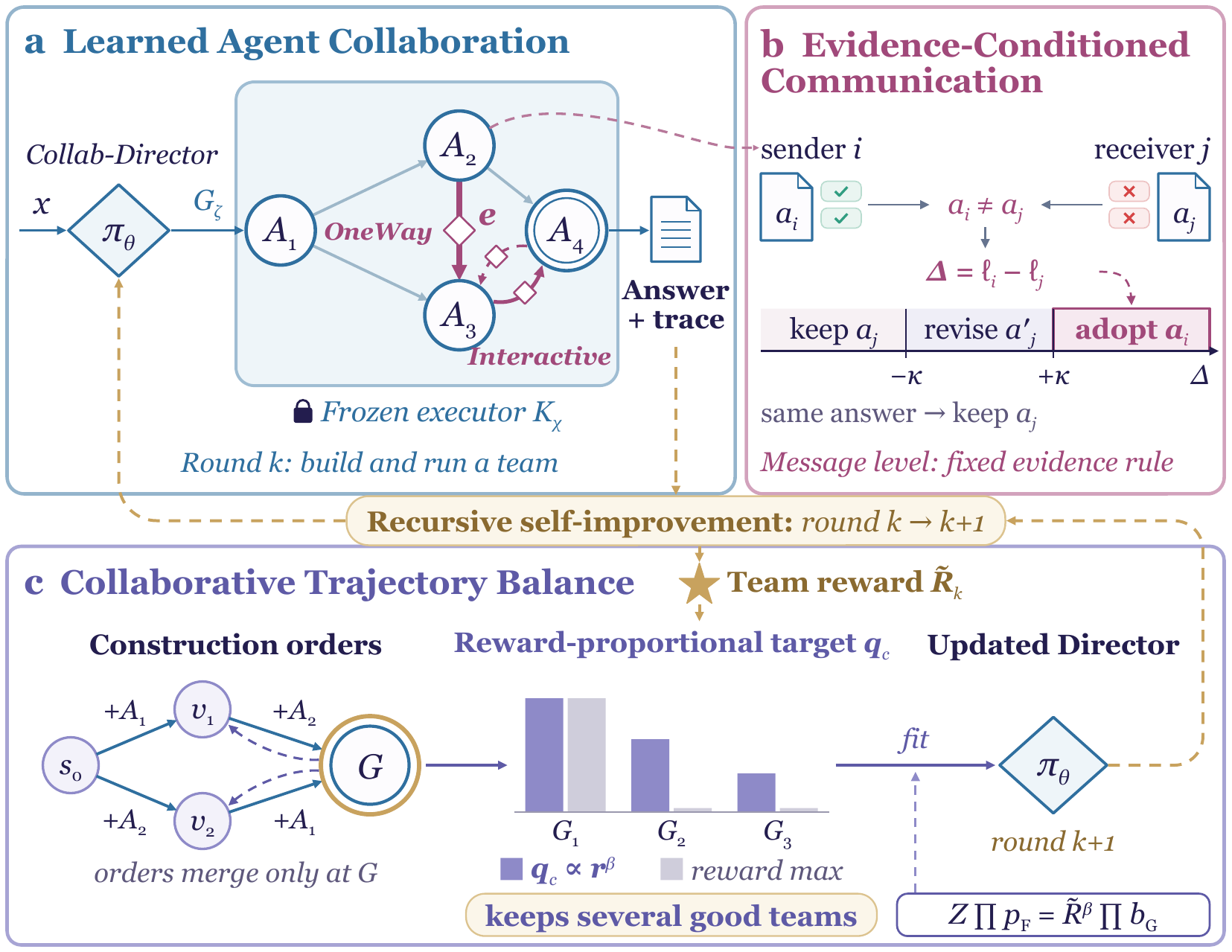}
\vspace{-10pt}
\caption{CollabFlow overview. (a) The director $\pi_\theta$ builds a team for executor $K_\chi$. (b) An evidence rule gates each message. (c) CTB retrains $\pi_\theta$.}
\label{fig:overview}
\vspace{-4pt}
\end{wrapfigure}

Recursive self-improvement (RSI), in which a system improves from its own outcomes, is becoming a concrete systems problem~\citep{chen2026recursive}. LLM-based multi-agent systems~\citep{li2023camel,wu2023autogen,hong2024metagpt} are a natural setting for it: their Agents already refine one another's answers, and every exchange leaves an execution outcome.

Such a system can improve at two levels, as shown in Figure~\ref{fig:overview}. Within a task, communication is the improvement step: specialized Agents supply complementary evidence and verify one another~\citep{qian2024chatdev,liu2024dylan}. Across tasks, the outcomes in turn improve how the team communicates, so the team level improves the message level.
In practice, however, teams still rely on hand-written roles, fixed protocols, and prompt-level rules~\citep{hong2024metagpt,qian2024chatdev}, which are costly to transfer across new tasks, Agent populations, or environments and which no execution outcome revises, so the loop from outcomes back to collaboration stays open.

\begin{figure}[t]
    \centering
    \includegraphics[width=\textwidth]{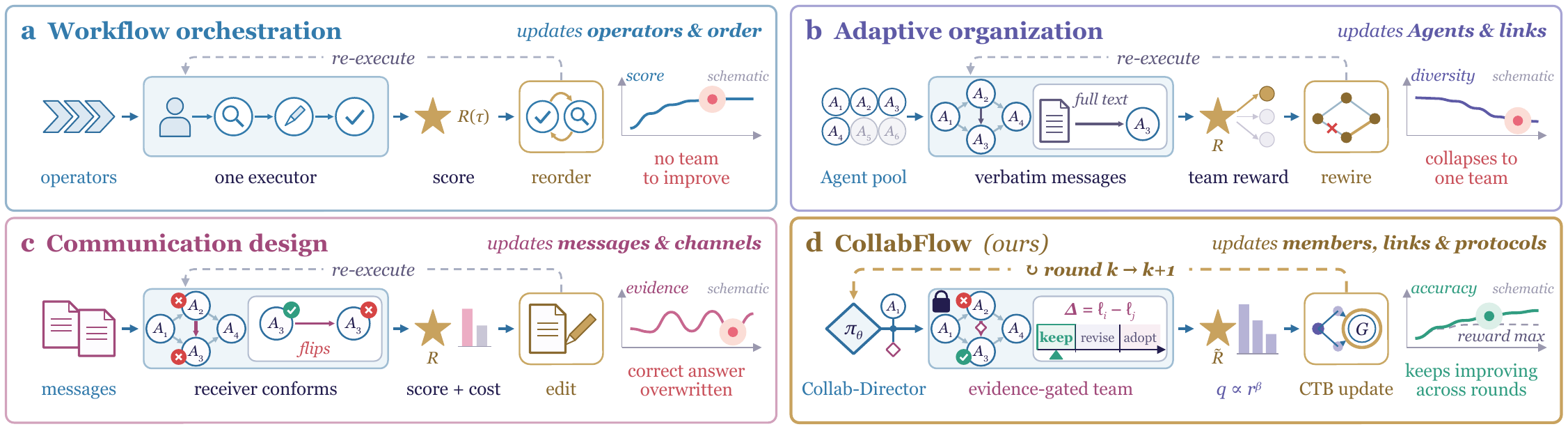}
    \caption{Multi-agent orchestration paradigms as self-improvement loops. Each paradigm runs the same loop of candidate designs, execution, outcome, feedback assessment, update, and re-execution, but updates a different object: (a) workflow orchestration updates operators and their order for one executor; (b) adaptive organization updates which Agents are used and linked; (c) communication design adjusts messages and channels; (d) CollabFlow jointly updates members, links, and edge protocols, trains the Collab-Director with CTB, and admits each message through an evidence gate.}
    \label{fig:taxonomy}
\end{figure}

To address these issues, three main paradigms of multi-agent orchestration have emerged, as shown in Figure~\ref{fig:taxonomy}.
First, workflow orchestration optimizes executable operator sequences for a single executor from execution feedback~\citep{zhang2025aflow,zhang2026flowsteer,zhang2026skillflow}.
Second, adaptive organization trains an orchestrator via policy gradient or group relative policy optimization (GRPO) to activate or connect Agents~\citep{dang2025evolving,tastan2026nexa,chen2026lemon,wang2026agentconductor}.
Third, communication design---such as debate, sparse consensus, multi-order messaging, and latent channels~\citep{du2023debate,gou2026dysco,guan2026moc,liu2026latentsurvey}---studies what information Agents exchange and how densely they interact within a task.

However, these methods leave the self-improvement loop open at three points:
\textbf{(i) Pre-defined, operator-level collaboration}---existing orchestration composes operators for a single executor~\citep{zhang2025aflow,zhang2026flowsteer} or fixes roles and message routes in advance~\citep{hong2024metagpt,dang2025evolving,tastan2026nexa}, so Agents are never the units that communicate;
\textbf{(ii) Topology-only learning under verbatim exchange}---learned topologies decide which Agents exchange messages~\citep{tastan2026nexa,yu2026codebook,jiang2026kgat}, but every received message re-enters the receiver's reasoning in full, spreading conformity and errors~\citep{ann2026interactiontax,bertalanic2026consensus,chun2026whatcommunicate} and letting concurrent writers corrupt the environment~\citep{yang2026concurrency}; the policy learns who talks but not how messages are used;
\textbf{(iii) Collapse of the team-level loop under reward maximization}---trained on its own outcomes, a reward-maximizing policy concentrates on a few teams~\citep{yu2026codebook,li2026dmpo}, leaving later rounds fewer teams to learn from, while the construction orders of one team split its credit.

To address these challenges, we propose \textbf{CollabFlow}, an RSI system of \emph{Learned Agent Collaboration}: a trainable Collab-Director builds teams of complete Agents, a frozen executor runs them, and each round's outcomes retrain the director.
Within a round, Agents refine one another's answers over a collaboration graph whose edges carry protocols of \emph{Evidence-Conditioned Communication}, which the Collab-Director selects together with the topology, so it learns who communicates and how.
On each edge, a receiver adopts a differing answer when the sender's evidence is stronger by a margin, keeps its own in the reverse case, and otherwise makes one bounded revision; on stateful tasks, a single Executor holds write permission.
Across rounds, the Collab-Director is trained via \emph{Collaborative Trajectory Balance} (CTB), a flow-based objective that credits each team once across its construction orders~\citep{malkin2022trajectory,bengio2023foundations} and targets a reward-proportional distribution over teams, so every positive-reward team keeps probability. Each round thus improves both the team and the records that define the next reward, and we bound how far this self-generated target moves between rounds, a bound that shrinks as the sampled teams' records accumulate.

We evaluate on twelve datasets spanning question answering, mathematical and expert reasoning, embodied planning, code generation, and web shopping, with held-out suites probing transfer to unseen task types.
CollabFlow outperforms direct-inference, fine-tuning, workflow-search, and agent-RL baselines in and out of distribution, lifts six other frozen executors, and uses the fewest tokens per item.
Across evolution steps, most of its gains persist and its team distribution stays diverse and stable, whereas reward-maximizing training shifts more between steps and finds fewer distinct solutions.
Ablations confirm each part: removing communication or director training costs the most, and every ablation costs more out of distribution than in it. As Figure~\ref{fig:taxonomy} shows, CollabFlow turns pre-defined teams into a collaboration that recursively improves itself under an explicit budget.

%% file: sections/02_related_work.tex
\section{Related Work}
\label{sec:related_work}

\textbf{Multi-Agent Collaboration and Self-Evolution.} Early multi-agent systems fixed roles and message routes by hand~\citep{hong2024metagpt,qian2024chatdev}. Learned orchestration now picks Agents sequentially~\citep{dang2025evolving}, predicts sparse topologies~\citep{tastan2026nexa,yu2026codebook}, or designs what Agents exchange~\citep{du2023debate,gou2026dysco}. In all three, messages are still exchanged verbatim by convention~\citep{motger2026madsurvey}, which induces conformity~\citep{ann2026interactiontax,bertalanic2026consensus}. Self-evolving systems adapt topologies at inference time through trace auditing~\citep{huang2026manta} or refine latent messages over recursion rounds~\citep{yang2026recursive}. CollabFlow makes team communication itself the object of recursive self-improvement: it jointly selects topology and protocols and retrains this choice on its own outcomes.

\textbf{Recursive Self-Improvement and RL for Agents.} Recursive self-improvement studies systems that produce better versions of themselves, mostly improving prompts, harnesses, skills, or weights against a fixed evaluator~\citep{chen2026recursive}. Agent RL is a common outer step of such loops, using trajectory rewards to sequence Agents~\citep{dang2025evolving}, edit workflow graphs~\citep{zhang2026flowsteer}, or evolve skills~\citep{zhang2026skillflow} under group-relative~\citep{deepseekmath2024} or counterfactual~\citep{chen2026lemon} credit. These objectives maximize reward and collapse to one mode, narrowing what later rounds learn from. Building on trajectory balance~\citep{malkin2022trajectory,bengio2023foundations}, CollabFlow scores the canonical complete-Agent team rather than each construction order.

%% file: sections/03_method.tex
\section{Preliminaries}
\label{sec:preliminaries}
\suppressfloats[t]
\begin{figure}[t]
 \centering
 \includegraphics[width=\textwidth]{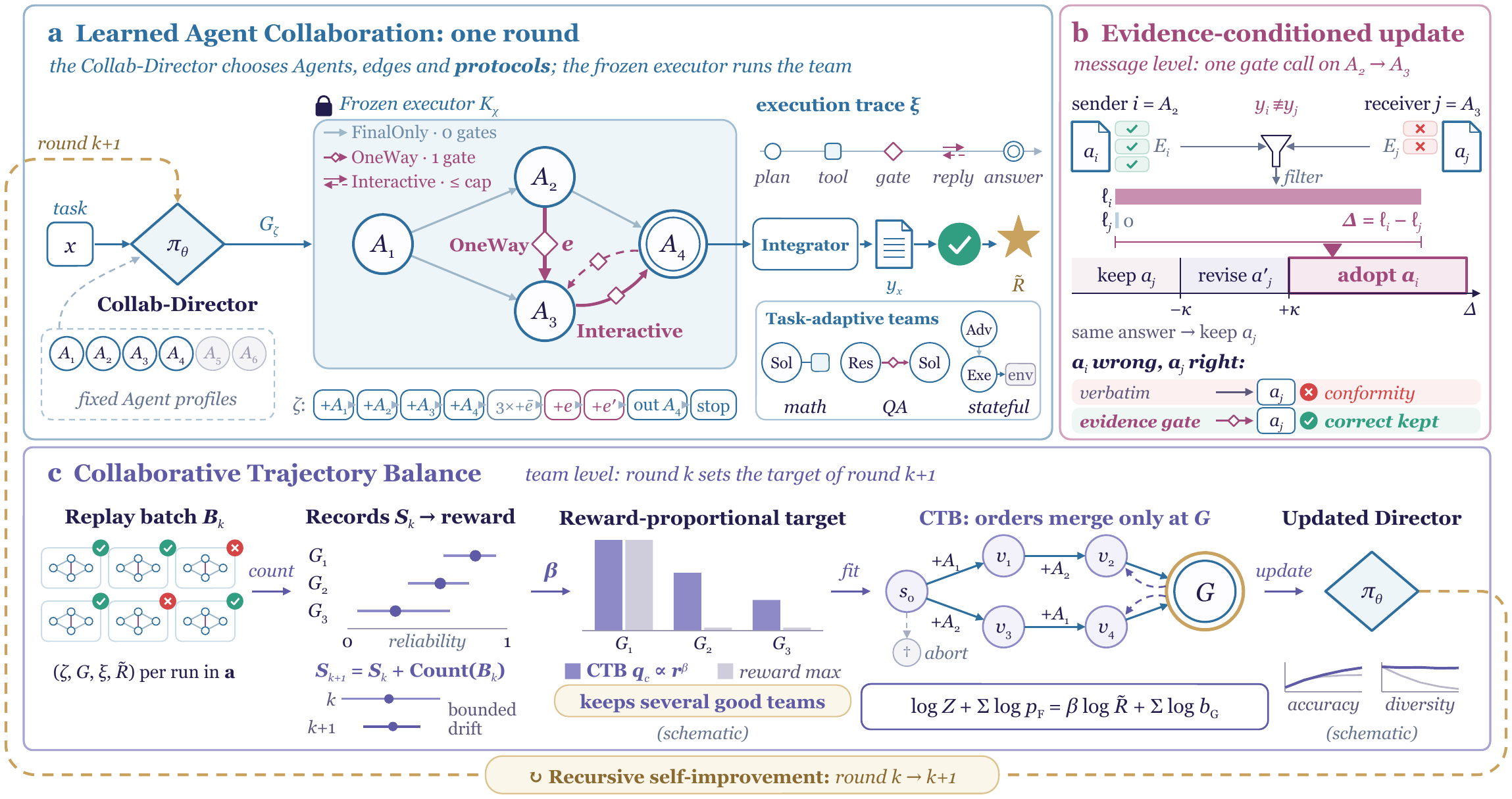}
 \caption{The self-improvement loop of CollabFlow. (a) Within a round, the Collab-Director builds a team from fixed Agent profiles through checked actions, the frozen executor runs it, and the evaluated outcome becomes a replay record (Section~\ref{subsec:paradigm}). (b) On each edge, the protocol sets how many evidence gates run, and a fixed evidence rule keeps, adopts, or revises the receiver's candidate (Section~\ref{subsec:graph}). (c) Across rounds, CTB merges construction histories only at the complete canonical team and updates the director with the length-normalized residual and a KL term (Section~\ref{subsec:rl}).}
 \label{fig:training_pipeline}
\end{figure}

\textbf{Definition 1: Collaboration Graph.} A collaboration graph is an annotated directed graph $G=(V,E,o)$, where $V\subseteq\mathcal A_x$ selects complete Agents with immutable profiles, $E\subseteq V\times V\times\mathcal P$ labels each edge with a protocol from the fixed registry $\mathcal P=\{\textsc{FinalOnly},\textsc{OneWay},\textsc{Interactive}\}$, and $o$ is the output rule. Agent profiles specify roles, tools, and inherited permissions, and stateful teams select exactly one authorized Executor. The loop rebuilds this graph in every round: its edges fix who communicates, and its protocols fix how each message may improve the receiver.

\textbf{Definition 2: Collaboration Trajectory.} Fix a context $c=(x,\chi,k)$ with task $x$, round $k$, and a specification $\chi$ of Agents, tools, contracts, and separate budgets $\mathbf B^{\rm con}+\mathbf B^{\rm exec}\le\mathbf B^{\rm total}$ (Appendix~\ref{app:theory_notation}). The Collab-Director builds each team through a construction trajectory on sufficient states $v_t$,
\begin{equation}
 \zeta=(v_0,e_1,v_1,\ldots,e_T,z),\qquad
 e_t=(v_{t-1},u_t,v_t),\qquad z\in\{g_G\}_{G\in\mathcal G_c}\cup\{\dagger\},
 \label{eq:orchestration_trajectory}
\end{equation}
where $T$ includes the accepted stop or the forced abort $\dagger$, and all accepted stops for one canonical team enter the same absorbing node $g_G$, so a valid trajectory yields $G_\zeta=\operatorname{Can}(\operatorname{Decode}(\zeta))$.

\textbf{Definition 3: Self-Improvement Trajectory.} Across rounds, CollabFlow updates the parameters $\Theta_k=(\theta_k,\phi_k,\psi_k)$ of the director, backward policy, and normalizer, and the reliability snapshot $\mathcal S_k=\{(s_k(G),n_k(G))\}_G$, from the batch $\mathcal B_k$ of round $k$:
\begin{equation}
 \Theta_{k+1}=\mathcal U_{\rm CTB}(\Theta_k;\mathcal B_k,\mathcal S_k),\qquad
 \mathcal S_{k+1}=\mathcal S_k+\operatorname{Count}(\mathcal B_k),
 \label{eq:rsi_trajectory}
\end{equation}
where $\operatorname{Count}$ adds each unique outcome once (Section~\ref{subsec:rl}), so the teams of round $k$ both retrain the director and define the reward of round $k+1$ through their success counts.

\textbf{Problem Statement.} Let $P_\theta$ be the law of construction attempts, $\mathsf V$ the event of valid completion, and $\mu_\theta^{\mathsf V}(G\mid c)=P_\theta(z=g_G\mid\mathsf V,c)$ the resulting team distribution. For a fixed positive reward table $r_G$ and $\beta>0$, each round pursues the team-level objective
\begin{equation}
 \min_\theta\ \mathbb E_{x\sim\mathcal D_X}
 D\big(\mu_\theta^{\mathsf V}(\cdot\mid c)\,\big\|\,q_c\big),\qquad
 q_c(G)\propto r_G^\beta,
 \label{eq:preliminary_objective}
\end{equation}
where $D$ is a divergence such as total variation, with context and reward fixed per round.

%% file: sections/04_methodology.tex
\section{Methodology: CollabFlow}
\label{sec:method}

As illustrated in Figure~\ref{fig:training_pipeline}, CollabFlow nests two improvement steps in the loop of Eq.~\ref{eq:rsi_trajectory}: within a round, Agents improve one another's answers under evidence-conditioned protocols, and across rounds, CTB retrains the Collab-Director on its own outcomes (Algorithm~\ref{alg:app_training}). Section~\ref{subsec:paradigm} defines the team, Section~\ref{subsec:graph} the message-level step, and Section~\ref{subsec:rl} the team-level step.

\subsection{Learned Agent Collaboration: What Each Round Changes}
\label{subsec:paradigm}

As shown in Figure~\ref{fig:training_pipeline}a, CollabFlow follows a Director--Executor paradigm: the trainable Collab-Director designs the team and a frozen executor runs its Agents, so each round changes only how teams are designed, from which Agents they use to how those Agents communicate.

\textbf{Team Construction.} A valid trajectory $\zeta$ yields a team $G_\zeta$, which the frozen executor $K_\chi$ runs:
\begin{equation}
 G_\zeta=\operatorname{Can}(G_0\oplus u_1\oplus\cdots\oplus u_T),\qquad
 (\xi,y_x)\sim K_\chi(\cdot\mid G_\zeta,x;\mathbf B^{\rm exec}).
 \label{eq:design-execute}
\end{equation}
where $\oplus$ applies checked updates, $\xi$ is the execution trace, and failures give $y_x=\bot$. Canonicalization merges histories that differ only in serialization order and keeps Agents, protocols, permissions, and outputs distinct, so a team's histories share its execution reservation and reward snapshot.

\textbf{Graph State and Valid Actions.} The state $s_t$ records the partial team and construction resources, and $v_t$ adds the accepted history and decision index. The director samples only from the legal set $\mathcal U_c(v_t)$ of actions that pass the construction checks at $v_t$, so no sampled proposal is rejected:
\begin{equation}
 \begin{gathered}
 s_t=(G_t,\mathcal A_x,\mathbf b_t^{\rm con}),\qquad v_t=(s_t,H_t,t),\\[2pt]
 \bar\pi_\theta(u\mid v_t,c)
 =\frac{\mathbf1\{u\in\mathcal U_c(v_t)\}\,\pi_\theta(u\mid v_t,c)}
 {\sum_{u'\in\mathcal U_c(v_t)}\pi_\theta(u'\mid v_t,c)}.
 \end{gathered}
 \label{eq:organization_state}
\end{equation}
Actions select Agents, bind skills, add protocol-labeled edges, select the output, or stop; a stop requires a complete valid team, and a timeout or exhausted budget returns an explicit abort (Appendix~\ref{app:action_semantics}).

\textbf{Marginalization and Decoupling.} Conditional on valid construction,
\begin{equation}
 P(dy\mid c,\mathsf V)=
 \sum_{G\in\mathcal G_c}\mu_\theta^{\mathsf V}(G\mid c)
 K_\chi(dy\mid G,x;\mathbf B^{\rm exec}).
 \label{eq:marginal-likelihood}
\end{equation}
Across rounds, only $\mu_\theta^{\mathsf V}$ changes in Eq.~\ref{eq:marginal-likelihood}, so every change in the system's output distribution comes from the teams the director builds, not from updates to the executor or the Agents.

\begin{proposition}
 \label{prop:construction}
 Under checked preservation (C1), every accepted trajectory yields a valid canonical team, and under component completeness and feasible prefixes (C2--C3), every target team keeps a legal construction trajectory in every round (Appendix~\ref{app:construction_soundness}).
\end{proposition}

\subsection{Message-Level Improvement: Evidence-Conditioned Communication}
\label{subsec:graph}

As shown in Figure~\ref{fig:training_pipeline}b, communication is the inner improvement step: within a round, Agents revise one another's answers. Verbatim exchange lets every message re-enter the receiver in full, which the DeGroot model~\citep{degroot1974reaching} treats as fixed-weight averaging whose consensus is set by the initial answers rather than by evidence (Lemma~\ref{lem:cf_degroot}). CollabFlow instead lets checkable evidence decide how each message is used, and the director learns how many evidence gates each edge runs.

\textbf{Evidence-Gated Update.} Agent $v$ ends its local episode with a candidate version $a_v$ and evidence records $E_v$. Its score $\ell_v=\max\{w(e):e\in\operatorname{Filter}_c(E_v,a_v)\}\cup\{0\}$ keeps the strongest record that still checks this version, with $w=2$ for a passed executable check and $w=1$ for retrieval support (Appendix~\ref{app:cf_provenance}). A directed update $i\to j$ between differing answers compares the evidence gap $\Delta=\ell_i-\ell_j$ of the sender over the receiver with a margin $\kappa$:
\begin{equation}
 a_j^{+}=a_i\ \ \text{if }\Delta>\kappa,\qquad
 a_j^{+}=a_j\ \ \text{if }\Delta<-\kappa,\qquad
 a_j^{+}=\mathrm{Revise}_c(a_j,a_i)\ \ \text{if }|\Delta|\le\kappa,
 \label{eq:gate}
\end{equation}
where $\mathrm{Revise}_c$ is one bounded, tool-free call in place of a re-execution, and its output is filtered and rescored, so every record stays with the version it checked. With integer scores in $\{0,1,2\}$, we fix $\kappa=1$, the only integer margin that adopts on a two-level gap but revises on a one-level gap. Then a sender whose answer passed a unit test ($\ell_i=2$) overrides a receiver whose answer is only asserted ($\ell_j=0$), while two asserted answers trigger a revision; stated confidence never enters $\Delta$.

\textbf{Who Communicates and How.} A protocol $p$ fixes the edge program $g^{p}_{e}$ that edge $e$ runs: the identity for \textsc{FinalOnly}, one gate for \textsc{OneWay}, and a capped number of gates in alternating directions for \textsc{Interactive}, stopping at agreement or a decisive update. One fixed-order sweep over the edges of $G$, topological for acyclic teams and repeated up to a fixed cap for cyclic ones, passes the candidates $\mathbf a=(a_v)_{v\in V}$ through the team's message-level operator
\begin{equation}
 \mathcal C_G=g^{p_{|E|}}_{e_{|E|}}\circ\cdots\circ g^{p_1}_{e_1},\qquad
 \mathbf a^{+}=\mathcal C_G(\mathbf a),
 \label{eq:message_operator}
\end{equation}
so the director's edges decide who communicates, their protocols decide how, and the recursion of Section~\ref{subsec:rl} improves both (Appendix~\ref{app:cf_runtime}). On stateful tasks, only the Executor may write and Advisors act through its next decision, so internal writes are serial~\citep{papadimitriou1979serializability}.

\textbf{Measurable Harm.} Equation~\ref{eq:gate} also fixes how a correct receiver can be hurt. Let \mbox{$C=\{z_j=1,z_i=0\}$}, \mbox{$J=\{|\Delta|\le\kappa\}$}, the false-adoption rate $\delta_C=P(\Delta>\kappa\mid C,c)$, and the revision error $c_J=P(z_j^{+}=0\mid J,C,c)$. On $C$ the answers differ, so the three branches partition $C$:
\begin{equation}
 P(z_j^{+}=0\mid C,c)=\delta_C+P(J\mid C,c)\,c_J.
 \label{eq:flip_bound}
\end{equation}
Choosing $\kappa=1$ over $\kappa=0$ removes $P(\Delta=1\mid C,c)$ from $\delta_C$ by sending one-level gaps to revision, and the same partition gives the rate at which a gate corrects a wrong receiver (Corollary~\ref{cor:cf_correction}); RQ4 ablates each of these branches. A flipped answer lowers its team's answer score and posterior success rate, so the next round's target moves toward teams whose messages help (Section~\ref{subsec:rl}).

\begin{proposition}
 \label{prop:revision}
 With version-bound evidence, ordered gating, and a single writer (S1--S3), evidence-conditioned communication keeps each record with the version it checked, flips a correct receiver only by false adoption or revision error, and serializes internal writes (Appendices~\ref{app:cf_provenance}, \ref{app:revision_bound}, and~\ref{app:state_serializability}).
\end{proposition}

\subsection{Team-Level Improvement: Collaborative Trajectory Balance}
\label{subsec:rl}

As shown in Figure~\ref{fig:training_pipeline}c, the team-level step retrains the Collab-Director across rounds on the outcomes of the teams it built and must keep several good teams in play. CTB builds on generative flow networks~\citep{bengio2021flow,bengio2023foundations} and trajectory balance (TB)~\citep{malkin2022trajectory}, with the complete-Agent team as its scored object: all construction orders of a team end at one canonical terminal, so the team receives its reward once, read from the loop's own records.

\textbf{Path Kernels.} The backward kernel decides how a team's credit is shared among the construction orders that reach it~\citep{shen2023towards}. Let $e_t=(v_{t-1},u_t,v_t)$ be a complete construction transition. With deterministic validation, the forward kernel is the legal-set policy of Eq.~\ref{eq:organization_state}, $p_F(e_t\mid v_{t-1},c)=\bar\pi_\theta(u_t\mid v_{t-1},c)$. The learned backward kernel $b_\phi(e_t\mid v_t,c)$ normalizes over the actual incoming parent--action pairs in the $G$-ancestral subgraph. Interior states carry their history and thus have one parent (Lemma~\ref{lem:cf_dag}), so $b_\phi$ acts only at the canonical terminal, where it picks one of the construction orders $\mathcal O(G)$ of $G$ by removing one component at a time; each removal is normalized, so $b_\phi$ sums to one over these orders (Appendix~\ref{app:likelihoods}). The uniform kernel $1/|\mathcal O(G)|$, by contrast, would give all construction orders of a team the same share of its credit.

\textbf{Balance and Training Surrogate.} In round $k$, for a fixed positive graph reward $r_G$, ideal balance asks every valid complete path $\zeta$ that ends at $g_G$ to satisfy
\begin{equation}
 \log Z_\psi(c)+\sum_{t=1}^{T_\zeta}\log p_F(e_t\mid v_{t-1},c)
 =\beta\log r_G+\sum_{t=1}^{T_\zeta}\log b_\phi(e_t\mid v_t,c).
 \label{eq:balance}
\end{equation}
In product form, $Z_\psi(c)$ times the forward probability of $\zeta$ equals $r_G^\beta$ times the backward probability of this order. Since backward probabilities sum to one over the orders of $G$, summing gives $Z_\psi(c)P_\theta(z=g_G\mid c)=r_G^\beta$, which fixes each team's probability and leaves its split over orders to $b_\phi$. Let $\rho_{\rm TB}(\zeta)$ be the left minus the right side of Eq.~\ref{eq:balance}. For the fixed replay batch $\mathcal B$ collected in round $k$ and its declared state weighting $d_{\mathcal B}$, CTB minimizes a length-normalized loss with KL weight $\alpha_{\rm KL}$
\begin{equation}
 \mathcal L(\theta,\phi,\psi)=
 \mathbb E_{\zeta\sim\mathcal B}(\rho_{\rm TB}(\zeta)/T_\zeta)^2
 +\alpha_{\rm KL}\mathbb E_{v\sim d_{\mathcal B}}
 D_{\rm KL}(\bar\pi_\theta(\cdot\mid v)\|\pi_{\rm ref}(\cdot\mid v)).
 \label{eq:organization_objective}
\end{equation}
Dividing the residual by path length puts short and long constructions on one scale without changing the zero-residual solutions. With $\pi_{\rm ref}=\bar\pi_{\theta_k}$, the director that collected the batch, the KL term is a proximal penalty as in TRPO and PPO~\citep{schulman2015trust,schulman2017proximal}: it damps each update, while every balanced director stays a fixed point (Lemma~\ref{lem:cf_proximal}).

\textbf{Reward-Proportional Teams.} At exact valid-path balance with fixed rewards,
\begin{equation}
 \mu_\theta^{\mathsf V}(G\mid c)=\frac{r_G^\beta}{W_{\mathsf V}},\qquad
 W_{\mathsf V}=\sum_{G'}r_{G'}^\beta,\qquad
 Z_\psi(c)P_\theta(\mathsf V\mid c)=W_{\mathsf V}.
 \label{eq:team_distribution}
\end{equation}
At this fixed point, every positive-reward team keeps probability, so the next round starts from a spread of teams. For example, with $\beta=1$ and rewards $0.6$, $0.4$, and $0.2$, the director samples the three teams with probabilities $1/2$, $1/3$, and $1/6$, the purple bars of Figure~\ref{fig:overview}c, whereas reward maximization places all mass on the first. Under inexact balance, a residual of magnitude at most $\eta$ on every valid path keeps $\mu_\theta^{\mathsf V}$ within total variation $\tanh(\eta/2)$ of $q_c$ (Theorem~\ref{thm:cf_approx}).

\textbf{Self-Generated Team Reward.} The reward of each round is partly generated by the loop itself. With counters frozen before batch $k$, an execution $\xi$ of team $G$ on task $x$ supplies
\begin{equation}
 \widetilde R_k(G,x,\xi)=\varepsilon_{\min}+r_{\rm ans}(G,x,\xi)
 \frac{s_k(G)+1/2}{n_k(G)+1}\exp[-L_{\rm rel}(G)].
 \label{eq:terminal_reward}
\end{equation}
Here $\varepsilon_{\min}>0$ keeps $\log\widetilde R_k$ finite, $0\le r_{\rm ans}\le1$, and $0\le s_k\le n_k$; the fixed structural code, measured in nats, and the reference class define $L_{\rm rel}\ge0$. In words, the reward multiplies the answer score by two factors with no tuned weight: the team's Jeffreys-smoothed posterior success rate~\citep{jeffreys1946invariant} and a minimum-description-length structural prior~\citep{rissanen1978modeling}, so past records decide which teams the next round favors. Counters add each unique outcome once after the batch, so each round generates the target of the next, and the more records a team has, the less one round can move its target; with fresh $\widetilde R_k$, the ideal $r_G$ is the team's geometric-mean reward $\exp\mathbb E[\log\widetilde R_k]$, which favors teams that succeed consistently (Theorem~\ref{thm:cf_random_reward}).

\begin{proposition}
 \label{prop:terminal_distribution}\label{prop:bounded_target}
 Under exact valid-path balance at fixed rewards (T1--T3, T4$'$), CTB samples each canonical team in proportion to its tempered reward among valid completions, and with a fixed answer-score table and structural code, the targets of consecutive rounds differ by a total-variation bound that shrinks as the sampled teams' records accumulate (Theorem~\ref{thm:cf_raw_valid}, Corollary~\ref{cor:cf_target_drift}).
\end{proposition}

%% file: sections/04_experiments.tex
\providecommand{\sd}[1]{\,{$_{\pm\text{#1}}$}}
\providecommand{\dl}[1]{\,\textcolor{deltapink}{(#1)}}
\definecolor{deltapink}{HTML}{D6337F}
\makeatletter
\providecommand{\fittab}[2]{%
  \sbox\z@{\tabcolsep=\z@ #2}%
  \dimen@=\textwidth \advance\dimen@ by -\wd\z@
  \divide\dimen@ by \numexpr 2*#1-2\relax
  \ifdim\dimen@<0.5pt \dimen@=0.5pt\fi
  \tabcolsep=\dimen@ #2}
\makeatother
\makeatletter\providecommand\captionof[1]{\def\@captype{#1}\caption}\makeatother

\section{Experiments}
\label{sec:experiments}

Six research questions test the claims of Sections~\ref{sec:introduction} and~\ref{sec:method}, asking whether learned collaboration wins in distribution, out of distribution and across frozen executors (RQ1--RQ3), and what the paradigm and components, the team-level objective and the recursion contribute (RQ4--RQ6).
\label{sec:experimental_setup}
\begin{table*}[t]
\centering
{\fontsize{6.9}{8.3}\selectfont
\tabcolsep=1pt
\renewcommand{\arraystretch}{1.32}
\fittab{11}{\begin{tabular}{@{}l@{\hspace{2pt}}l|cccccc|cc|c@{}}
\toprule
 & & \textbf{Baseline} & \textbf{SFT} & \textbf{GRPO} & \textbf{AFlow} & \multicolumn{2}{c|}{\textbf{Agent+RL}} & \multicolumn{2}{c|}{\textbf{Skill evolution}} & \textbf{Ours} \\
\cmidrule(lr){3-3}\cmidrule(lr){4-4}\cmidrule(lr){5-5}\cmidrule(lr){6-6}\cmidrule(lr){7-8}\cmidrule(lr){9-10}\cmidrule(lr){11-11}
\textbf{Dataset} & \textbf{Metric} & \textbf{Qwen3.5} & \textbf{Qwen3.5} & \textbf{Qwen3.5} & \textbf{Qwen3.5} & \textbf{FlowSteer} & \textbf{Evolving} & \textbf{SkillFlow} & \textbf{SkillRL} & \textbf{CollabFlow ($\Delta\uparrow$)} \\
\midrule
\multicolumn{11}{@{}l}{\textit{(a) In-Distribution (IID) benchmarks}} \\
\midrule
\textbf{HotpotQA} & EM & 46.25\sd{0.86} & 47.34\sd{0.89} & 49.53\sd{1.62} & 53.75\sd{1.78} & 53.28\sd{2.02} & 60.16\sd{2.14} & 60.78\sd{0.86} & 58.28\sd{2.04} & \textbf{63.91}\sd{1.69}\dl{+17.7} \\
 & F1 & 65.59\sd{0.89} & 67.32\sd{1.46} & 67.80\sd{2.02} & 69.97\sd{1.50} & 71.38\sd{2.18} & 72.95\sd{1.05} & 78.37\sd{1.95} & 75.71\sd{1.96} & \textbf{79.72}\sd{1.96}\dl{+14.1} \\
\textbf{TriviaQA} & EM & 69.69\sd{1.28} & 69.53\sd{1.66} & 69.84\sd{1.96} & 87.34\sd{1.60} & 87.34\sd{1.02} & 86.88\sd{1.16} & 86.56\sd{1.50} & 82.97\sd{0.35} & \textbf{87.97}\sd{1.88}\dl{+18.3} \\
 & F1 & 74.67\sd{2.11} & 77.03\sd{2.10} & 76.62\sd{1.52} & 89.04\sd{0.43} & 88.55\sd{1.87} & 89.71\sd{1.10} & 89.41\sd{1.18} & 86.28\sd{1.62} & \textbf{91.23}\sd{2.14}\dl{+16.6} \\
\textbf{AIME 2026} & Acc. & 49.33\sd{1.49} & 47.33\sd{2.79} & 49.33\sd{2.79} & 58.67\sd{2.98} & 66.67\sd{2.36} & 70.67\sd{1.49} & 71.33\sd{1.83} & 51.33\sd{2.98} & \textbf{75.33}\sd{1.83}\dl{+26.0} \\
\textbf{MedXpertQA} & Acc. & 65.78\sd{0.86} & 68.59\sd{1.69} & 75.78\sd{1.10} & 95.47\sd{1.50} & 96.09\sd{1.83} & 95.62\sd{0.89} & 90.00\sd{1.69} & 89.84\sd{2.21} & \textbf{97.97}\sd{0.70}\dl{+32.2} \\
\textbf{ALFWorld} & SR & 33.59\sd{1.56} & 40.47\sd{2.10} & 43.44\sd{1.62} & 79.69\sd{1.46} & 93.59\sd{2.02} & 95.47\sd{1.87} & 92.81\sd{1.40} & 87.34\sd{1.02} & \textbf{98.44}\sd{1.99}\dl{+64.8} \\
\textbf{MBPP+} & Pass@1 & 74.53\sd{1.80} & 75.47\sd{2.57} & 76.56\sd{1.75} & 77.66\sd{1.52} & 78.12\sd{1.35} & 77.34\sd{1.91} & 81.41\sd{1.40} & 75.78\sd{1.24} & \textbf{84.84}\sd{1.80}\dl{+10.3} \\
\cmidrule(lr){1-11}
\multirow{3}{*}{\textbf{Avg.\,(IID)}} & EM & 57.97\sd{0.77} & 58.44\sd{0.94} & 59.69\sd{1.27} & 70.55\sd{1.20} & 70.31\sd{1.13} & 73.52\sd{1.22} & 73.67\sd{0.86} & 70.62\sd{1.03} & \textbf{75.94}\sd{1.27}\dl{+18.0} \\
 & F1 & 70.13\sd{1.15} & 72.18\sd{1.28} & 72.21\sd{1.27} & 79.50\sd{0.78} & 79.97\sd{1.44} & 81.33\sd{0.76} & 83.89\sd{1.14} & 81.00\sd{1.27} & \textbf{85.47}\sd{1.45}\dl{+15.3} \\
 & Acc. & 55.81\sd{0.73} & 57.97\sd{1.16} & 61.28\sd{0.96} & 77.87\sd{0.99} & 83.62\sd{0.96} & 84.78\sd{0.80} & 83.89\sd{0.79} & 76.08\sd{1.01} & \textbf{89.15}\sd{0.83}\dl{+33.3} \\
\midrule
\multicolumn{11}{@{}l}{\textit{(b) Out-of-Distribution (OOD) benchmarks}} \\
\midrule
\textbf{MuSiQue} & EM & 38.59\sd{2.57} & 38.75\sd{1.80} & 39.38\sd{1.05} & 41.41\sd{1.46} & 41.72\sd{1.52} & 44.53\sd{1.10} & 43.59\sd{1.40} & 39.84\sd{0.78} & \textbf{46.41}\sd{1.88}\dl{+7.8} \\
 & F1 & 38.69\sd{1.10} & 39.49\sd{1.40} & 39.08\sd{1.42} & 44.41\sd{0.86} & 43.84\sd{1.62} & 53.59\sd{2.02} & 46.23\sd{1.96} & 43.69\sd{1.66} & \textbf{59.93}\sd{2.21}\dl{+21.2} \\
\textbf{NQ-Open} & EM & 27.97\sd{1.87} & 30.78\sd{1.71} & 28.44\sd{2.04} & 75.78\sd{1.66} & 78.44\sd{0.89} & 88.28\sd{2.14} & 84.53\sd{1.60} & 81.72\sd{1.96} & \textbf{93.44}\sd{0.70}\dl{+65.5} \\
 & F1 & 49.41\sd{1.91} & 54.40\sd{1.05} & 48.49\sd{1.71} & 82.39\sd{1.40} & 84.85\sd{1.71} & 91.62\sd{1.95} & 88.27\sd{1.88} & 85.38\sd{1.62} & \textbf{93.73}\sd{1.16}\dl{+44.3} \\
\textbf{MATH-Hard} & Acc. & 77.81\sd{1.71} & 77.50\sd{1.02} & 71.72\sd{1.28} & 85.78\sd{0.86} & 93.44\sd{1.18} & 87.34\sd{1.40} & 85.31\sd{1.78} & 78.75\sd{0.86} & \textbf{94.69}\sd{1.95}\dl{+16.9} \\
\textbf{GPQA-D} & Acc. & 66.25\sd{1.28} & 70.31\sd{2.21} & 69.53\sd{1.46} & 74.69\sd{2.04} & 77.03\sd{1.31} & 73.75\sd{1.05} & 74.06\sd{1.95} & 73.28\sd{1.60} & \textbf{79.84}\sd{1.69}\dl{+13.6} \\
\textbf{HumanEval} & Pass@1 & 82.34\sd{1.88} & 80.00\sd{1.62} & 82.81\sd{0.55} & \textbf{86.41}\sd{1.18} & 84.22\sd{1.40} & 83.91\sd{1.42} & 83.75\sd{2.02} & 82.03\sd{1.83} & \textbf{86.41}\sd{1.80}\dl{+4.1} \\
\textbf{WebShop} & SR & 44.69\sd{1.40} & 46.88\sd{0.55} & 48.12\sd{1.52} & 65.94\sd{1.31} & 80.62\sd{1.78} & 78.28\sd{1.69} & 81.88\sd{1.28} & 76.41\sd{1.28} & \textbf{84.69}\sd{1.52}\dl{+40.0} \\
\cmidrule(lr){1-11}
\multirow{3}{*}{\textbf{Avg.\,(OOD)}} & EM & 33.28\sd{1.59} & 34.77\sd{1.24} & 33.91\sd{1.15} & 58.59\sd{1.10} & 60.08\sd{0.88} & 66.41\sd{1.20} & 64.06\sd{1.06} & 60.78\sd{1.06} & \textbf{69.92}\sd{1.00}\dl{+36.6} \\
 & F1 & 44.05\sd{1.10} & 46.95\sd{0.87} & 43.79\sd{1.11} & 63.40\sd{0.82} & 64.35\sd{1.18} & 72.61\sd{1.40} & 67.25\sd{1.36} & 64.53\sd{1.16} & \textbf{76.83}\sd{1.25}\dl{+32.8} \\
 & Acc. & 67.77\sd{0.79} & 68.67\sd{0.74} & 68.05\sd{0.63} & 78.20\sd{0.71} & 83.83\sd{0.72} & 80.82\sd{0.70} & 81.25\sd{0.89} & 77.62\sd{0.72} & \textbf{86.41}\sd{0.87}\dl{+18.6} \\
\bottomrule
\end{tabular}}
}
\vspace{2pt}
\caption{Main results on twelve IID and OOD benchmarks. SFT and GRPO are applied on Qwen3.5-9B; all trainable orchestration methods use Qwen3.5-9B as the frozen executor under the same per-task budget. $\Delta\uparrow$ is the gap to direct Qwen3.5-9B. Every cell is the mean over five runs; the subscript is the standard deviation across those five, propagated across benchmarks for the Avg.\ rows. QA averages are taken separately for EM and F1, and Acc.\ averages the primary metrics (accuracy, success rate, or Pass@1) of the other four benchmarks in each part.}
    \label{tab:main_results}
\end{table*}
\begin{table*}[!tp]
\centering
{\fontsize{6.6}{7.9}\selectfont
\tabcolsep=1pt
\renewcommand{\arraystretch}{1.08}
\fittab{11}{\begin{tabular}{@{}l@{\hspace{4pt}}|@{\hspace{4pt}}cccccc@{\hspace{4pt}}|@{\hspace{4pt}}cccccc@{}}
\toprule
\textbf{Variant} & \multicolumn{6}{c|}{\textbf{IID}} & \multicolumn{6}{c}{\textbf{OOD}} \\
\cmidrule(lr){2-7}\cmidrule(lr){8-13}
 & HotpotQA & TriviaQA & AIME & MedXpertQA & ALFWorld & MBPP+ & MuSiQue & NQ-Open & MATH & GPQA & HumanEval & WebShop \\
 & EM & EM & Acc. & Acc. & SR & Pass@1 & EM & EM & Acc. & Acc. & Pass@1 & SR \\
\midrule
Single agent with tools & 47.66 & 78.12 & 56.67 & 82.03 & 85.16 & 66.41 & 25.78 & 72.66 & 78.91 & 56.25 & 70.31 & 32.81 \\
Parallel vote & 58.59 & 85.94 & 73.33 & 96.88 & 89.06 & 82.03 & 42.19 & 89.84 & 92.97 & 75.78 & 82.81 & 50.78 \\
Fixed chain & 53.91 & 81.25 & 73.33 & 84.38 & 96.09 & 83.59 & 39.06 & 82.03 & 94.53 & 68.75 & 84.38 & 53.91 \\
Fully connected debate & 52.34 & 86.72 & 70.00 & 94.53 & 92.19 & 80.47 & 43.75 & 87.50 & 93.75 & 78.91 & 81.25 & 44.53 \\
Evolving orchestrator & 57.03 & 85.16 & 63.33 & 91.41 & 95.31 & 78.91 & 35.94 & 84.38 & 89.84 & 73.44 & 78.91 & 60.16 \\
\midrule
$-$ Director training & 50.78 & 78.91 & 53.33 & 89.84 & 82.81 & 71.88 & 28.91 & 66.41 & 73.44 & 64.84 & 75.78 & 58.59 \\
$-$ Recursion & 57.81 & 83.59 & 60.00 & 94.53 & 89.06 & 78.91 & 32.03 & 76.56 & 79.69 & 69.53 & 80.47 & 65.62 \\
Fixed team & 61.72 & 87.50 & 70.00 & 97.66 & 93.75 & 82.81 & 44.53 & 88.28 & 92.97 & 76.56 & 84.38 & 78.12 \\
Random legal graph & 58.59 & 84.38 & 63.33 & 93.75 & 90.62 & 79.69 & 38.28 & 82.03 & 88.28 & 71.88 & 82.81 & 72.66 \\
\midrule
$-$ Communication & 46.88 & 75.78 & 50.00 & 88.28 & 78.91 & 68.75 & 21.09 & 63.28 & 70.31 & 60.94 & 78.91 & 55.47 \\
$-$ Evidence gate & 57.03 & 86.72 & 63.33 & 96.09 & 92.19 & 81.25 & 30.47 & 72.66 & 78.12 & 67.97 & 81.25 & 68.75 \\
All \textsc{OneWay} & 55.47 & 82.81 & 60.00 & 94.53 & 87.50 & 77.34 & 27.34 & 75.00 & 75.78 & 70.31 & 80.47 & 62.50 \\
Always revise & 53.91 & 85.16 & 56.67 & 92.97 & 85.94 & 75.78 & 24.22 & 68.75 & 86.72 & 65.62 & 79.69 & 67.19 \\
\midrule
\textbf{CollabFlow (full)} & \textbf{63.28} & \textbf{88.28} & \textbf{76.67} & \textbf{98.44} & \textbf{97.66} & \textbf{85.16} & \textbf{46.88} & \textbf{92.97} & \textbf{95.31} & \textbf{79.69} & \textbf{85.94} & \textbf{84.38} \\
\bottomrule
\end{tabular}}
}
\vspace{-2pt}
\caption{Paradigm comparison and ablation of CollabFlow on the six IID and six OOD benchmarks (Qwen3.5-9B executor; EM for QA; one run). Paradigm rows fix how Agents collaborate before execution: one tool-using Agent, parallel voting, a fixed chain, fully connected debate, and the evolving orchestrator. Team-level rows remove director training or recursion across rounds, or replace the learned team with a fixed team or a random legal graph; message-level rows remove communication or the evidence gate, force every edge to \textsc{OneWay}, or revise on every update.}
    \label{tab:ablation}
    \vspace{2pt}
    \begin{minipage}[t]{0.695\textwidth}\vspace{0pt}
        \centering
        \includegraphics[width=\linewidth]{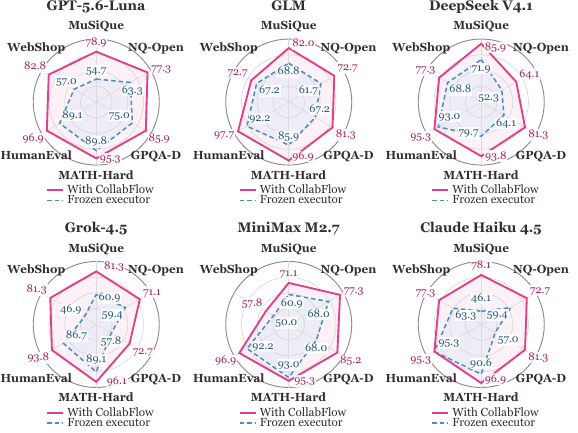}\\[1pt]
        {\small (a) OOD benchmark profile on each frozen executor}
    \end{minipage}\hfill
    \begin{minipage}[t]{0.295\textwidth}\vspace{0pt}
        \centering
        \includegraphics[width=\linewidth]{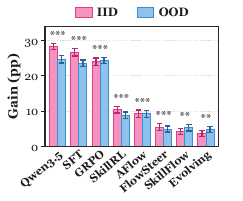}\\[-1pt]
        {\small (b) Significance of the gains}\\[3pt]
        \includegraphics[width=\linewidth]{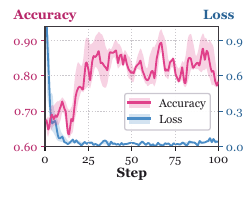}\\[-1pt]
        {\small (c) Training curve}
    \end{minipage}\\[0pt]
    \includegraphics[width=\textwidth]{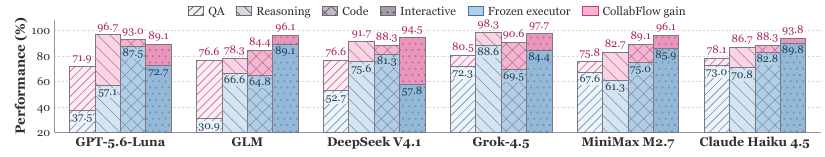}\\[-3pt]
    {\small (d) IID scores aggregated by task domain}\vspace{0pt}
    \captionof{figure}{Other frozen executors, significance of the gains, and training dynamics. \textbf{(a)} OOD scores of each frozen executor used directly (dashed) and with CollabFlow (solid); each axis is linear over its own score range (values in Table~\ref{tab:executors}). \textbf{(b)} Gain of CollabFlow over each baseline in the main average, five-run mean with standard-error bars; Welch $t$-test: *** $p<10^{-3}$, ** $p<10^{-2}$. \textbf{(c)} Training accuracy and CTB loss; bands give the spread over runs. \textbf{(d)} IID scores averaged per task domain (QA: HotpotQA, TriviaQA; Reasoning: AIME 2026, MedXpertQA; Code: MBPP+; Interactive: ALFWorld); blue bars give the frozen executor alone and pink bars the gain from CollabFlow.}
    \label{fig:executors}
\end{table*}

\textbf{Setup.} Six in-distribution (IID) datasets train the Collab-Director and six out-of-distribution (OOD) datasets are held out, 128 items each (30 for AIME 2026) (Table~\ref{tab:main_results}; \citealp{wu2026spatialscore,hotpotqa,triviaqa,medxpertqa,alfworld,musique,naturalquestions,gpqa,humaneval}). Baselines are direct Qwen3.5-9B, SFT and GRPO~\citep{deepseekmath2024} on it, AFlow~\citep{zhang2025aflow}, FlowSteer~\citep{zhang2026flowsteer}, the evolving orchestrator~\citep{dang2025evolving}, SkillFlow~\citep{zhang2026skillflow} and SkillRL~\citep{xia2026skillrl}. Qwen3.5-9B serves as both the Collab-Director and the frozen executor, and all methods share one Agent population, tool set and per-task budget, which also charges messages. We report exact match (EM, primary) and F1 for QA and accuracy, success rate or Pass@1 elsewhere; Table~\ref{tab:main_results} and Figure~\ref{fig:executors}b use five runs with Welch's $t$-test, and Figure~\ref{fig:executors}a,d uses six other executors (Appendices~\ref{app:datasets}--\ref{app:implementation}).

\textbf{In distribution (RQ1).} CollabFlow leads every IID dataset and metric (Table~\ref{tab:main_results}), 2.27 EM and 4.37 accuracy points above the strongest baseline on average, and every average gain is significant at $p<10^{-2}$ (Figure~\ref{fig:executors}b). SFT and GRPO add at most 5.5 accuracy points to the backbone, whereas learned teams add 33.3: as Eq.~\ref{eq:marginal-likelihood} states, with the executor and Agents frozen, every gain comes from the teams the director builds and how their members exchange evidence.

\input{sections/fig_mechanism}

\textbf{Out of distribution (RQ2).} With everything frozen, CollabFlow again leads every OOD row but HumanEval, where it ties AFlow, and its margins widen to 3.51 EM and 4.22 F1, while fine-tuning moves the backbone by under one accuracy point. The director learns which roles to combine and how messages are used rather than task answers, and held-out tasks also yield checkable evidence.

\textbf{Across executors (RQ3).} CollabFlow raises every IID and OOD average of six other frozen executors (Figure~\ref{fig:executors}a,d; Table~\ref{tab:executors}), and weaker executors gain more (Pearson $r=-0.96$ in distribution), narrowing the IID spread across executors from 23.3 to 9.3 points: gated teams of complete Agents supply the checks that a weak executor lacks alone.

\textbf{Paradigms and components (RQ4).} Every paradigm and ablated variant falls below the full model on all twelve datasets (Table~\ref{tab:ablation}), each losing more out of distribution; a single tool-using Agent is lowest, and fully connected debate trails parallel voting. Removing director training costs 13.7 IID and 19.5 OOD points, of which recursion across rounds supplies 7.6 and 13.6. Removing communication costs most, while verbatim exchange, forced revision and all-\textsc{OneWay} edges cost 5.5--9.8 IID but 14.3--15.6 OOD points: by Proposition~\ref{prop:revision}, a correct receiver flips only by false adoption or revision error, and each variant widens one route, most under shift, where more candidates err.

\textbf{Team-level objective (RQ5).} Under matched rollouts and reward, CTB beats GRPO, PPO and CollabFlow without CTB on every team-level measure, with the smallest team-distribution shift (Figure~\ref{fig:mechanism}b). Its reward-proportional target (Eq.~\ref{eq:team_distribution}) keeps every positive-reward team in play and Proposition~\ref{prop:bounded_target} bounds how far that target moves, so CollabFlow finds 26 distinct successful paths against at most 15 (Figure~\ref{fig:mechanism}a), and training accuracy passes 0.80 within 25 steps (Figure~\ref{fig:executors}c).

\textbf{Recursive self-improvement (RQ6).} Between evolution steps, the loop edits the skills the director binds to Agents (Appendix~\ref{app:profile_schema}). Over the twenty steps with the largest one-step gain, success on the evolved task rises from 34.4\% to 75.0\%, 15 steps stay ahead five steps later, and they carry 83.6\% of the 55 edits (Figure~\ref{fig:mechanism}c--e), as the reward of Eq.~\ref{eq:terminal_reward} counts each outcome once and smooths it with the team's record. It also uses the fewest tokens per item with the lowest spread (Figure~\ref{fig:mechanism}f; Table~\ref{tab:token_cost}).

%% file: sections/fig_mechanism.tex
\makeatletter
\newsavebox\mechTabA \newsavebox\mechTabB \newlength\mechGap
\newcommand\mechRow{\rule{0pt}{\dimexpr\ht\@arstrutbox+\mechGap\relax}}
\makeatother
\newcommand\mechTableA{{\fontsize{6.4}{7.6}\selectfont\renewcommand{\arraystretch}{1.12}\setlength{\tabcolsep}{2.6pt}%
\begin{tabular}{@{}l|c|ccc|cc@{}}
\toprule
 & & \multicolumn{3}{c|}{\textbf{Successful paths}} & \multicolumn{2}{c}{\textbf{Successful tool use}} \\
\cmidrule(lr){3-5}\cmidrule(l){6-7}
 & SR & Distinct & Path & SW path & Tool & Tool comb. \\
\textbf{Method} & (\%)\,$\uparrow$ & paths\,$\uparrow$ & disp.\,$\uparrow$ & diff.\,$\uparrow$ & trans.\,$\uparrow$ & entropy\,$\uparrow$ \\
\midrule
FlowSteer & 35.33 & 13 & 0.9156 & 0.0640 & 13 & 1.4878 \\
SkillFlow & 86.30 & 14 & 0.7541 & 0.4146 & 15 & 1.4635 \\
SkillOpt & 86.30 & 14 & 0.7572 & 0.4128 & 15 & 1.4675 \\
Untrained director & 78.44 & 15 & 0.7234 & 0.5619 & 16 & 1.6460 \\
\midrule
\textbf{CollabFlow} & \textbf{90.19} & \textbf{26} & \textbf{0.9253} & \textbf{0.5743} & \textbf{18} & \textbf{1.7033} \\
\bottomrule
\end{tabular}}}
\newcommand\mechTableB{{\fontsize{6.4}{7.6}\selectfont\renewcommand{\arraystretch}{1.12}\setlength{\tabcolsep}{1.6pt}%
\begin{tabular}{@{}l|ccc|cc|cc@{}}
\toprule
 & \multicolumn{3}{c|}{\textbf{Reward and success}} & \multicolumn{2}{c|}{\textbf{Consistency}} & \multicolumn{2}{c}{\textbf{Stability}} \\
\cmidrule(lr){2-4}\cmidrule(lr){5-6}\cmidrule(l){7-8}
\textbf{Team-level} & Avg. & Succ. & 4/4 & Reward & Worst & Succ. & Team \\
\textbf{step} & reward\,$\uparrow$ & (\%)\,$\uparrow$ & (\%)\,$\uparrow$ & std\,$\downarrow$ & task\,$\uparrow$ & std\,$\downarrow$ & TV\,$\downarrow$ \\
\midrule
\mechRow GRPO & 0.8226 & 85.16 & 70.31 & 0.1107 & 67.19 & 8.38 & 0.1631 \\[\mechGap]
\mechRow PPO & 0.8336 & 86.72 & 70.31 & 0.1191 & 75.00 & 8.66 & 0.1506 \\[\mechGap]
\mechRow w/o CTB & 0.8578 & 88.28 & 81.25 & 0.0496 & 78.13 & 9.47 & 0.1491 \\[\mechGap]
\midrule
\mechRow \textbf{CollabFlow} & \textbf{0.8775} & \textbf{92.58} & \textbf{85.94} & \textbf{0.0353} & \textbf{85.94} & \textbf{6.04} & \textbf{0.0937} \\[\mechGap]
\bottomrule
\end{tabular}}}
\begin{figure*}[t]
\centering
\setlength{\mechGap}{0pt}%
\sbox\mechTabA{\mechTableA}\sbox\mechTabB{\mechTableB}%
\setlength{\mechGap}{\dimexpr(\ht\mechTabA+\dp\mechTabA-\ht\mechTabB-\dp\mechTabB)/8\relax}%
\sbox\mechTabB{\mechTableB}%
\typeout{MECHTAB A=\the\dimexpr\ht\mechTabA+\dp\mechTabA\relax\space B=\the\dimexpr\ht\mechTabB+\dp\mechTabB\relax\space gap=\the\mechGap}%
\begin{minipage}[b]{0.545\textwidth}
\centering\usebox\mechTabA\\[2pt]
{\small (a) Diversity of successful solutions}
\end{minipage}\hfill
\begin{minipage}[b]{0.435\textwidth}
\centering\usebox\mechTabB\\[2pt]
{\small (b) Team-level step on the first 64 prompt groups}
\end{minipage}\\[5pt]
\includegraphics[width=\textwidth]{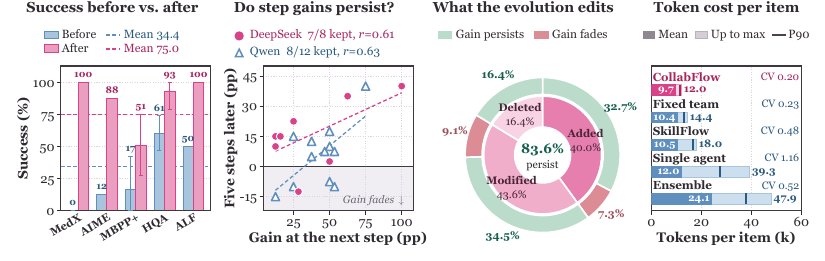}\\[-1pt]
{\small\makebox[0.25\textwidth]{\hspace{0.12in}(c) Evolution success}\makebox[0.25\textwidth]{\hspace{0.06in}(d) Gain persistence}\makebox[0.25\textwidth]{\hspace{0.22in}(e) Skill edits}\makebox[0.25\textwidth]{\hspace{0.06in}(f) Token cost}}
\vspace{-3pt}
\caption{Solution diversity, team-level training, skill evolution and cost. \textbf{(a)} Successful solutions on 1,438 items, with SkillOpt~\citep{yang2026skillopt}; SW: success-weighted. \textbf{(b)} First 64 prompt groups of training (8 steps, 256 rollouts per method); TV between consecutive team distributions. \textbf{(c)} Success over 8 rollouts before and after the twenty evolution steps with the largest one-step gain; dashed: means. \textbf{(d)} Gain at the next step against five steps later. \textbf{(e)} The 55 skill edits by type (inner) and by whether their step still gains five steps later (outer). \textbf{(f)} Tokens per item.}
\label{fig:mechanism}
\end{figure*}

%% file: sections/05_conclusion.tex
\section{Conclusion}
\label{sec:conclusion}

CollabFlow makes agent collaboration recursively self-improving: evidence-conditioned protocols let Agents refine one another's answers within a round, and Collaborative Trajectory Balance retrains the Collab-Director on its own records across rounds while keeping several good teams in play, and the target it generates moves between rounds by a bound that shrinks as records accumulate. On twelve datasets it leads every baseline in and out of distribution, beats every fixed paradigm and ablated variant, lifts six other frozen executors, keeps most of its step gains, and uses the fewest tokens per item (limitations and future work in Appendices~\ref{app:limitations} and~\ref{app:future_work}).

\section*{AI Use Statement}
In this work, generative AI tools were used to edit human-written text for clarity and readability, assist with \LaTeX{} typesetting and page layout, and adjust the placement and sizing of existing figures and tables. They were not used to generate or modify experimental data or scientific-figure content, generate synthetic data, formulate the research hypotheses or mathematical claims, develop the theoretical or conceptual framework, design the methodology or experiments, implement the proposed method, clean or reformat datasets, perform qualitative data analysis, or interpret the experimental results. The authors manually reviewed and verified all AI-assisted text, figures, numbers, citations, and formatting against the implementation and experimental records. The authors take full responsibility for the final content of the paper and all AI-assisted artifacts.

\section*{Ethics Statement}
This work follows the principles of the ICLR Code of Ethics. We aim to conduct and report our research with scientific integrity, transparency, and reproducibility. Experimental results are reported without fabrication, falsification, or intentional misrepresentation, and every reported number is produced by the released code. We appropriately acknowledge prior work and use existing public benchmarks, models, in accordance with their intended purposes and licenses. CollabFlow's Agents act only through permitted tools, a single Executor holds write permission on stateful tasks, and results on medical questions do not authorize clinical use (Appendix~\ref{app:applicability}). This study does not involve human participants or the collection of personal information.

\section*{Reproducibility Statement}
We release the code at \url{https://anonymous.4open.science/r/CollabFlow-631E}, including the training and evaluation code, the trained Collab-Director weights, and the training and test splits of all twelve datasets. Section~\ref{sec:method} specifies the team construction, the evidence-gated update and the CTB objective, and Section~\ref{sec:experiments} specifies the benchmarks, baselines, shared budget and five-run protocol. The appendix states the assumptions and proofs of the three propositions (Appendix~\ref{app:theory}), the Collab-Director interface and Agent population (Appendix~\ref{app:interfaces}), the training procedure (Appendix~\ref{app:algorithms}), the datasets and evaluation protocol (Appendix~\ref{app:datasets}), baselines and ablations (Appendix~\ref{app:baselines}), metrics (Appendix~\ref{app:metric_details}), and per-benchmark results with other executors and token costs (Appendix~\ref{app:additional_results}).

%% file: sections/appendix.tex
\numberwithin{equation}{section}
\numberwithin{table}{section}
\numberwithin{algorithm}{section}

\providecommand{\appTBD}{\textemdash}
\providecommand{\appPending}[1]{\textit{[TBD: #1]}}

\section{Collab-Director Prompts, Agent Population, and Runtime Interfaces}
\label{app:interfaces}

This appendix supplements the collaboration graph, the communication
semantics, and the training procedure in the main text.
Appendix~\ref{app:theory} states the assumptions of the three
propositions and proves them; the remaining sections specify the
algorithms and evaluation protocol.

\subsection{Collab-Director Input and Output Contract}
\label{app:supervisor_contract}

The Collab-Director receives the task $x$, compatible Agent profiles
$\mathcal A_x$, the partial graph $G_t$, the protocols permitted for the task
family, and the remaining construction budget $\mathbf b_t^{\rm con}$.
Its context $H_t$ also contains the accepted construction decisions.
Reference answers, hidden evaluation tests, and evaluator-only labels are
not exposed as construction inputs.
Feedback produced by a permitted tool during ordinary execution is
distinct from hidden evaluator information.

The Collab-Director selects one structural action at a time.
The following is a \emph{schematic prompt contract}, not a verbatim prompt
recovered from a training run:

\begin{quote}
Construct an Agent team for the supplied task using the available
profiles.
Select one currently legal construction action.
Choose Agents, typed communication relations, and an output mechanism;
leave task solving to the executor.
Respect compatibility, write-permission, and resource constraints.
Stop only when the team is valid.
Do not invent Agent identifiers, and do not use evaluator-only
information.
\end{quote}

The frozen system prompt, user template, structured-action schema, and
invalid-action correction prompt are released with the code.

\subsection{Agent Profiles}
\label{app:profile_schema}

An Agent is a complete execution role rather than an individual operator.
Table~\ref{tab:app_profile_schema} lists the profile fields that the
Collab-Director can read when constructing a team.
Static profile capability tags describe what an Agent may do; they are
not modified during training.
Agent profile definitions and executor parameters remain fixed; the
Collab-Director changes which profiles are selected, which skills they
carry, and how their instances communicate.

\begin{table}[!htbp]
    \centering
    \small
    \caption{
        Conceptual schemas used by the method.
        Exact registry identifiers and serialization are
        configuration-dependent.
    }
    \label{tab:app_profile_schema}
    \begin{tabularx}{\linewidth}{
        @{}p{0.24\linewidth}X@{}
    }
        \toprule
        \textbf{Object / field}
        & \textbf{Meaning}
        \\
        \midrule

        Agent identity and role
        & Stable profile identifier, role description,
          and local objective.
        \\

        Compatibility
        & Supported task families and capability labels.
        \\

        Tool permissions
        & Permitted tools and whether an operation can modify
          protected environment state.
        \\

        Visibility
        & Task evidence, messages, and observations visible
          to the Agent.
        \\

        Execution mode
        & Stateless Agent, mutable \textsc{Executor},
          or read-only \textsc{Advisor}.
        \\

        \bottomrule
    \end{tabularx}
\end{table}

The construction interface also supports a
\texttt{bind\_skill} action that attaches a retrieved textual strategy to
one selected Agent; it grants no tool permissions and never overrides
write permission.
Between evolution steps, the loop adds, modifies or deletes library
skills from the outcomes of the teams that used them; Figure~\ref{fig:mechanism}c--e
reports these edits.

\subsection{Structural Actions and Validation}
\label{app:action_semantics}

Table~\ref{tab:app_actions} expands the action space of Section~\ref{subsec:paradigm}.
Its arguments are mathematical objects.
The legal finite action registry is exposed through a constrained schema.
Schema legality and probability normalization are separate contracts:
the director scores every legal action string and samples from the
renormalized kernel of Eq.~\ref{eq:organization_state}, whereas a
grammar-masked token sampler would use its own normalized likelihood
(Appendix~\ref{app:cf_likelihood_contract}).
The runtime checks each sampled action before changing the graph, so every
sampled action is accepted; scoring tokens, latency, and model calls count
toward the construction budget, and a timeout follows the declared abort rule.

\begin{table*}[!htbp]
    \centering
    \small
    \caption{
        Structural-action semantics and local validity conditions.
        Terminal constraints are checked separately from
        partial-state invariants.
    }
    \label{tab:app_actions}
    \begin{tabularx}{\textwidth}{
        @{}p{0.20\textwidth}p{0.27\textwidth}X@{}
    }
        \toprule
        \textbf{Action}
        & \textbf{Structural effect}
        & \textbf{Required conditions}
        \\
        \midrule

        \texttt{add\_agent}$(a)$
        & Add $a$ to $V$.
        & $a\in\mathcal A_x\setminus V$;
          profile, permission, capacity, and budget checks pass.
        \\

        \texttt{add\_edge}$(a_i,a_j,p)$
        & Add $(a_i,a_j,p)$ to $E$.
        & Both endpoints have been selected;
          $p\in\mathcal P$ is supported for the task and endpoint
          modes; the edge is not already present; relation and budget checks pass.
        \\

        \texttt{bind\_skill}$(a,k)$
        & Attach library skill $k$ to Agent $a$.
        & $a\in V$; the skill grants no tool permission and never
          overrides write permission.
        \\

        \texttt{set\_output}$(o')$
        & Set output configuration to $o'$.
        & The output is not yet set; the configuration is supported,
          all referenced Agents exist, and its execution is compatible
          with task permissions and budget.
        \\

        \texttt{stop}
        & Commit the terminal organization.
        & The team is nonempty, output is defined, and all
          task-specific terminal constraints are satisfied.
        \\

        \bottomrule
    \end{tabularx}
\end{table*}

Partial-state validity permits an unfinished graph.
For example, the output configuration can be unset before termination,
and a stateful partial team may have \emph{at most} one mutable Executor,
whereas a completed stateful team requires \emph{exactly} one.
Enforcing every terminal condition on the empty initial graph would
prevent construction.
Candidate resources are held fixed within a construction episode;
a changed population defines a different resource context.

\subsection{Communication and Output Semantics}
\label{app:communication_semantics}

\begin{table}[!htbp]
    \centering
    \small
    \caption{
        Communication semantics.
        Bounds and task-specific availability are set by the
        frozen runtime configuration.
    }
    \label{tab:app_protocols}
    \begin{tabularx}{\linewidth}{
        @{}p{0.22\linewidth}X@{}
    }
        \toprule
        \textbf{Protocol}
        & \textbf{Execution meaning}
        \\
        \midrule

        \textsc{FinalOnly}
        & Make a completed result available without rerunning
          the receiver.
        \\

        \textsc{OneWay}
        & Send evidence or a proposal to the receiver, which may
          revise once under the configured resource bound.
        \\

        \textsc{Interactive}
        & Allow a bounded sequence of mutual correction rounds
          where the task and runtime support it.
        \\

        Stateful restriction
        & Internal Advisors provide bounded one-way guidance to
          the sole mutable Executor, and write permission stays
          with the Executor.
        \\

        \bottomrule
    \end{tabularx}
\end{table}

The execution model first produces stateless candidates and then processes
communication under the fixed schedule of Appendix~\ref{app:cf_runtime}.
The resulting revision trace records actual messages and responses rather
than inferring communication from graph edges alone.
For organization-class statistics, a graph with multiple Agents but no
enabled intermediate communication is treated as independent, even when
final results are routed to an output component.
Report both planned edges and messages actually executed.

Within \textsc{OneWay} and \textsc{Interactive}, $\mathrm{Revise}$ is a single
model call with tools disabled and at most $\tau_{\mathrm{rev}}$ tokens, so an
one visit to an edge program costs at most $r_{\max}$ revision calls of
$\tau_{\mathrm{rev}}$ tokens. A cyclic schedule may visit the program up to
$R_{\rm comm}$ times; all such calls are charged to its execution quota.
A round denotes one elementary directed gate evaluation. Revision never
re-executes the receiver's tool episode. Each evidence item is bound to a
candidate version and carries a provenance identifier. The context supplied
to \textsc{Revise} may contain the union of peer evidence, but after a
revision only evidence that remains valid for the returned candidate is kept
and duplicate identifiers are removed, so a revised candidate earns its own
verification score.
Under the stateful restriction an Advisor's message is a structured handoff
(state, evidence, next subgoal, completion check, recovery) produced without
tools, and edges are restricted to \textsc{Advisor}$\to$\textsc{Executor}
\textsc{OneWay}.

A \textsc{Single} output selects one Agent result.
An Integrator deduplicates records, filters them for its chosen candidate
representative, and applies the declared score rule of
Section~\ref{subsec:graph}; a Verifier evaluates candidates
under its configured contract, and a Formatter changes submission format
without resolving the task again.
\textsc{Single} and Formatter are counted as single-output modes, and the
configuration declares which output modes are available.

\subsection{Execution Records and Information Boundaries}
\label{app:trace_schema}

An episode records the construction trajectory $\zeta$, canonical graph
$G_\zeta$, execution trace $\xi$, output $\hat y$, factual score
$r_{\mathrm{ans}}$, and resource ledger.
Construction decisions are the organization-policy learning targets.
Child-Agent messages, tool results, and environment observations are
execution evidence, separate from the Collab-Director's action tokens.
The evaluator receives the information required for factual scoring,
while inference-time Agents receive only task-permitted inputs.

For reproducibility, records should include the task/split identifier,
resource versions, accepted and rejected proposals, communication protocol
and round, sender/receiver identifiers, before/after revisions, tool
outcomes, stopping reason, and measured costs.
The exact logging keys are given in the released code.
This separation also prevents a post-execution reference answer from
being reintroduced as a Collab-Director input.

\input{sections/appendix_theory}

\section{Algorithmic Details}
\label{app:algorithms}

\subsection{Training Procedure}
\label{app:training_algorithm}

Algorithm~\ref{alg:app_training} expands the self-improvement trajectory of
Eq.~\ref{eq:rsi_trajectory} and the procedure specified in
Section~\ref{subsec:rl}.
Each round keeps the Agent population and executor fixed, collects several team
rollouts for each sampled task, and updates the director, backward
policy, and normalizer; its unique outcomes then update the reliability
counters used by the next round.

The released configuration fixes batching, retry handling, and replay.

\begin{algorithm}[!ht]
\caption{CollabFlow self-improvement rounds: raw valid-only replay}
\label{alg:app_training}
\begin{algorithmic}[1]
\Require Tasks, fixed resource registry and executor, separate
$\mathbf B^{\rm con},\mathbf B^{\rm exec}$, rollout count $N$.
\Ensure Updated parameters $(\theta,\phi,\psi)$ and versioned records.
\For{each self-improvement round $k$}
  \State Freeze the resource context, pre-batch counters, and reward-code version.
  \State Initialize empty replay batch $\mathcal B$ and unique-outcome buffer.
  \For{each sampled task $x$ and rollout $i=1,\ldots,N$}
    \State Construct by sampling from the legal-action kernel of
    Eq.~\ref{eq:organization_state}; log every transition, stop, and abort
    with its log probability.
    \If{construction returns a valid canonical team $G_i$}
      \State Enter its canonical absorbing terminal with the full legal parent--action support.
      \State Reset the execution instance and use the fixed execution reservation.
      \State Run the version-filtered protocols under the fixed bounded schedule;
      obtain $(\hat y_i,\xi_i)$, including declared execution failure.
      \State Compute observed reward $\widetilde R_k(G_i,x,\xi_i)$ using the
      pre-batch snapshot; retain the snapshot key.
      \State Add the complete valid path and observed reward to $\mathcal B$;
      buffer the unique execution outcome for the later counter update.
    \Else
      \State Record construction abort; do not relabel it as a valid team.
    \EndIf
  \EndFor
  \If{$\mathcal B\ne\varnothing$}
    \State Evaluate current forward log probabilities and the backward log
    probability $\log b_\phi$ of each replayed construction order.
    \State Optimize Eq.~\ref{eq:organization_objective} with
    $\pi_{\rm ref}=\bar\pi_{\theta_k}$, treating the replay distribution and
    recorded reward targets as fixed during the update.
  \EndIf
  \State Update counters once per unique buffered outcome, after the batch.
\EndFor
\end{algorithmic}
\end{algorithm}

This pseudocode specifies raw valid-only replay. Its ideal fixed-reward,
all-valid-path zero-residual interpretation (T4$'$) is Theorem~\ref{thm:cf_raw_valid}:
$ZP_\theta(\mathsf V)=W_{\mathsf V}$ and only the valid-conditional graph law
is reward-proportional. Appendix~\ref{app:cf_training_limits} treats fresh
execution rewards and the KL term. Empty replay skips the loss update;
failed execution of a valid graph remains a valid graph sample with its
observed reward.

At inference, freeze the Collab-Director, construct a legal team, and
execute it with the frozen executor.
Ordinary environment observations may influence an executing Agent's
local actions within the task.

\subsection{Reliability Statistics and Reward Bookkeeping}
\label{app:reward_details}

For a fixed graph/context key with $s$ successes in $n$ unique outcomes,
conditionally independent Bernoulli outcomes with one fixed success parameter
and a Beta$(1/2,1/2)$ prior give
\begin{equation}
 p\mid\text{records}\sim\operatorname{Beta}(s+1/2,\,n-s+1/2).
 \label{eq:app_beta_posterior}
\end{equation}
The posterior mean under this model is
\begin{equation}
 \widehat p_k(G)=\frac{s_k(G)+1/2}{n_k(G)+1},
 \label{eq:app_reliability_mean}
\end{equation}
with counts keyed by the graph and its task cluster $c_x$.
It is the posterior mean under the Bernoulli model of
Lemma~\ref{lem:cf_reward_support}, and event deduplication keeps a copied
record from counting twice.
Task-cluster pooling uses a fixed documented key. Binary success counts are
separate from a continuous factual score such as F1.

At update $k$, all rollouts use pre-batch counts, so each execution is
scored against the records that preceded it. The unique-outcome buffer updates
counts only after the batch. Each record binds the task key, graph identity,
resource and runtime versions, execution reservation, factual score, and
snapshot identifier. A fresh execution reward is denoted
$\widetilde R_k(G,x,\xi)$; the fixed $r_G$ in a theorem is a different,
explicitly fixed graph functional or immutable table.

For a fixed graph code length $L_\chi$ and class map $a_\chi$, define
$L_{\rm rel}(G)=L_\chi(G)-\min_{G':a_\chi(G')=a_\chi(G)}L_\chi(G')$.
The nonempty reference class, code, and minima are fixed before the reward
snapshot. This definition makes the structural factor nonnegative in
length and at most one after exponentiation with $\lambda\ge0$ (the method
uses $\lambda=1$);
Appendix~\ref{app:cf_reward_snapshot} proves support and counter-drift bounds.
The proof holds for any specified finite code.

\subsection{Action Likelihoods, Masks, and Backward Terms}
\label{app:likelihoods}

For a canonical serialized structural action with tokens
$(y_1,\ldots,y_m)$ and context $H$, the autoregressive sequence
log-probability is
\begin{equation}
\log p_\theta(u\mid H)
=
\sum_{j=1}^{m}
\log p_\theta(y_j\mid H,y_{<j}).
\label{eq:app_action_likelihood}
\end{equation}

The action's representation includes any required completion delimiter.
When the legal policy is a renormalization of raw action probabilities,
its probability is obtained by dividing by the total mass of legal
actions.
A schema-constrained sampler or rejection sampler records its own
probability convention, so the stored log score equals its normalized
sampling probability.

Only intended Collab-Director action targets, including termination,
contribute to organization-action likelihoods.
Task text, runtime feedback, child-Agent outputs, and evaluator labels
are context or auxiliary evidence.
If sampled reflection text affects an action, its probability must
either be included in the modeled path or be handled by an explicitly
defined conditional/marginal model.
Report the exact mask, serialization, reduction, and reference-policy
conventions.

The action log probability is the sum of its token log probabilities,
which normalizes over actions; this is the probability that
Proposition~\ref{prop:terminal_distribution} uses. Log-space evaluation
avoids underflow and keeps this definition.

\paragraph{Backward kernel.} Under legal-set sampling, the complete paths that
reach the canonical terminal $g_G$ are exactly the construction orders
$\mathcal O(G)$ of $G$: the prerequisite-respecting orders of its Agent and edge
insertions and its output selection, followed by \texttt{stop}. The learned
kernel $b_\phi$ scores an order by unbuilding $G$: starting from $G$, it removes
one component whose removal leaves a legal prefix (an edge, the output
selection, or an Agent that no remaining edge or output references) through a
softmax over these removable components, and multiplies the choices. Every step
is normalized, so $b_\phi(\cdot\mid G)$ sums to one over $\mathcal O(G)$, which
is T3 at the terminal; interior history states have one parent and backward
probability one. As in TB \citep{malkin2022trajectory}, $b_\phi$ is learned;
\citet{shen2023towards} analyze how the backward policy shares credit among
paths. The uniform ablation sets
$b(\zeta\mid G)=1/|\mathcal O(G)|$,
where $|\mathcal O(G)|$ is computed exactly by dynamic programming over subsets
of the components of $G$.

\subsection{Trajectory Contrasts and Diagnostic Evidence}
\label{app:diagnostic_details}

The transition statistic is the reward-weighted occupancy
$H_c(a)=\sum_\zeta r_{G_\zeta}^\beta Q_{G_\zeta}(\zeta)N_a(\zeta)$ over
valid paths, with abort contributions set to zero. Appendix~\ref{app:cf_occupancy}
defines the state and transition occupancies, the endpoint-conditioned
forward bridge, the on-target frequency estimator, and the importance-weighted
estimator for a declared proposal. No conditional average over invoking
trajectories is substituted for an unconditional occupancy. At finite
training error, raw forward frequencies describe the learned policy.

For scored paths, store $d_t=\log p_F-\log b_G$ and the complete-path
identity in Eq.~\ref{eq:cf_path_diagnostic}. Cached probabilities can be reused;
new backward evaluations, parent normalizers, sampling, and arithmetic
remain chargeable computation.

An optional same-task contrast uses two recorded runs under the same context,
orders them by the recorded factual score (ties are retained as ties), and
reports their longest identical sequence of accepted structural actions,
first divergence, subsequent messages, evidence versions, costs, and outcomes.
All available distinct record pairs may be reported as descriptive error
analysis.

\subsection{Structural Complexity and Cost Accounting}
\label{app:complexity}

For a partial team with $n$ selected Agents, $A$ compatible candidates,
$P$ communication protocols, and at most $C_o$ output configurations, a conservative candidate-action
bound is
\begin{equation}
|\mathcal U(s)|
\le
(A-n)+Pn^2+nK+C_o+1.
\label{eq:app_action_bound}
\end{equation}

Binding one of $K$ library skills to a selected Agent adds the $nK$ term. The bound permits
all endpoint pairs and therefore remains valid when
self-links or additional duplicates are subsequently excluded.
A naive complete enumerator over $T$ decisions consequently inspects
at most
\[
O\!\left(
T(A+Pn^2+nK+C_o+1)
\right)
\]
candidates when using the maximum $n$ along the trajectory.

These are structural counting bounds on construction.
End-to-end cost must additionally include Agent execution, revisions,
tool calls, and evaluator work, using the accounting conventions in Appendix~\ref{app:metric_details}.

\section{Dataset and Evaluation Protocol Details}
\label{app:datasets}

\subsection{Benchmark Inventory and Splits}

Table~\ref{tab:app_datasets} lists the twelve benchmarks, their metrics,
and the number of evaluation items; the released code contains the
training and test splits.
IID suites keep their evaluation examples out of training.

\begin{table*}[!htbp]
    \centering
    \small
    \caption{Benchmarks, metrics, and evaluation items.}
    \label{tab:app_datasets}
    \begin{tabularx}{\textwidth}{@{}p{0.22\textwidth}p{0.08\textwidth}Xp{0.16\textwidth}c@{}}
        \toprule
        \textbf{Benchmark} & \textbf{Suite} & \textbf{Task} & \textbf{Metric} & \textbf{Eval. items} \\
        \midrule
        HotpotQA & IID & Multi-hop question answering & EM / F1 & 128 \\
        TriviaQA & IID & Factual question answering & EM / F1 & 128 \\
        AIME 2026 & IID & Mathematical reasoning & Accuracy & 30 \\
        MedXpertQA-Text & IID & Medical reasoning & Accuracy & 128 \\
        ALFWorld & IID & Embodied task completion & Success rate & 128 \\
        MBPP+ & IID & Function-level code generation & Pass@1 & 128 \\
        \midrule
        MuSiQue & OOD & Multi-hop question answering & EM / F1 & 128 \\
        NQ-Open & OOD & Open-domain factual answers & EM / F1 & 128 \\
        MATH-Hard & OOD & Mathematical reasoning & Accuracy & 128 \\
        GPQA Diamond & OOD & Scientific reasoning & Accuracy & 128 \\
        HumanEval & OOD & Function-level code generation & Pass@1 & 128 \\
        WebShop & OOD & Web shopping in a simulated store & Success rate & 128 \\
        \bottomrule
    \end{tabularx}
\end{table*}

Train and test identifiers are disjoint, including derived prompts and
duplicated source instances.
Dataset citations remain those in the main paper's bibliography.

\subsection{Task Adapters and Factual Scoring}

QA adapters specify supplied versus retrieved evidence, retrieval
access, answer extraction, and accepted reference aliases.
Mathematical and expert-reasoning adapters specify the answer schema,
parsing failures, and correctness matcher.
Code-generation adapters execute submitted functions against their
permitted test harness and record the test outcome.
Embodied and conversational adapters expose their respective
observation/action interfaces and terminal evaluators.

HumanEval and MBPP+ use separate submission formats, and the
transfer from ALFWorld to WebShop changes the interaction interface as
well as the task text.
Adapter revisions, timeouts, and retry policies are given in the released code.
Evaluator-only labels and hidden tests must not enter Agent prompts.

\subsection{Frozen Held-out Evaluation}
\label{app:ood_protocol}

The OOD protocol freezes the Collab-Director, Agent population, and
training-derived configuration.
No OOD evaluation instance is used for policy training, prompt
selection, or hyperparameter tuning.
Benchmark-specific adapters were fixed before outcome inspection.
The six OOD comparisons include within-family changes and stronger
domain or interface shifts.

For WebShop, the single-Executor restriction covers CollabFlow's
internal Agents only.
The benchmark-controlled web server is external, and the Advisor is
consulted between bounded execution phases with a fresh page snapshot.
We record both policy-controlled browser actions and server-side state.

\section{Baselines, Ablations, and Comparison Controls}
\label{app:baselines}

\subsection{Baseline Inventory}

The main comparison covers direct Qwen3.5-9B inference, SFT, GRPO,
AFlow, FlowSteer, the evolving orchestrator, SkillFlow, and SkillRL;
Table~\ref{tab:app_baselines} lists the configuration of each.

\begin{table*}[!htbp]
    \centering
    \small
    \caption{
        Baselines and their configurations.
    }
    \label{tab:app_baselines}
    \begin{tabularx}{\textwidth}{
        @{}p{0.19\textwidth}p{0.27\textwidth}X@{}
    }
        \toprule
        \textbf{Method}
        & \textbf{Role in comparison}
        & \textbf{Configuration}
        \\
        \midrule

        Direct Qwen3.5-9B
        & Shared-backbone reference
        & Checkpoint, tools, prompt, decoding and submission limits.
        \\

        SFT
        & Task-level supervised baseline
        & Demonstration source, training split, trainable
          components, optimization budget.
        \\

        GRPO
        & Task-level RL baseline
        & Policy target, reward, sample group, optimizer
          and rollout budget.
        \\

        AFlow
        & Workflow-search baseline
        & Search budget, backend, candidate components
          and task adapter.
        \\

        FlowSteer
        & Workflow-orchestration baseline
        & Policy/backend versions, action interface, reward
          and adapter changes.
        \\

        Evolving orchestrator
        & Adaptive-organization baseline
        & Orchestrator policy, Agent pool, reward
          and adapter changes.
        \\

        SkillFlow
        & Orchestration/Skill baseline
        & Policy, Skill lifecycle, workspace initialization
          and adapter changes.
        \\

        SkillRL
        & Policy/Skill baseline
        & Training data, Skill initialization/evolution,
          backend and adapter changes.
        \\

        CollabFlow
        & Proposed full specification
        & Released code and configuration;
          all components enabled.
        \\

        \bottomrule
    \end{tabularx}
\end{table*}

Where interfaces permit, match backbone, tool access,
token/model-call/tool-call limits, and interaction horizon.
Report unmatched settings explicitly rather than calling all systems
identical.
Include output aggregation and formatting in the ledger.
Search, sampling, and training budgets are separate quantities; one
must not substitute for another when attributing a difference to the
organization policy.

\subsection{Fixed Teams and Random Legal Graphs}

The structural controls are a single Agent, independent parallel
Agents, a fixed researcher--solver--verifier chain, a stateful
Executor--Advisor team, and randomly sampled legal graphs.
A fixed topology is selected using training or validation data and
reused within its task family; no per-test-instance oracle chooses
the topology.

The random-graph proposal and stopping distribution are recorded,
because they define the distribution over terminal graphs.

Independent Agent outputs require a specified aggregation mechanism.
An ensemble that votes is different from running several Agents and
selecting one unchanged answer.
Fix the output mechanism when comparing communication, and charge all
aggregation computation.
Controls requiring three Agents report their search space.

\subsection{Component Ablations}

Table~\ref{tab:app_ablation_protocol} defines the paradigm and
ablation rows of Table~\ref{tab:ablation}; every row keeps the frozen
executor, Agent population, and per-task budget of the full model.

\begin{table*}[!htbp]
    \centering
    \small
    \caption{Paradigm and ablation rows of Table~\ref{tab:ablation}.}
    \label{tab:app_ablation_protocol}
    \begin{tabularx}{\textwidth}{@{}p{0.26\textwidth}X@{}}
        \toprule
        \textbf{Row} & \textbf{Intervention} \\
        \midrule
        Single agent with tools & One tool-using Agent answers alone. \\
        Parallel vote & Independent Agents answer; a majority vote selects the output. \\
        Fixed chain & A fixed researcher--solver--verifier chain with verbatim messages. \\
        Fully connected debate & Every Agent receives every other Agent's full output each round. \\
        Evolving orchestrator & The orchestrator of \citet{dang2025evolving} activates Agents in sequence. \\
        \midrule
        $-$ Director training & The untrained Collab-Director builds the teams. \\
        $-$ Recursion & The director is trained once and not retrained on later rounds. \\
        Fixed team & The best hand-built team per task family, chosen on training data. \\
        Random legal graph & Teams are sampled uniformly from the legal action set. \\
        \midrule
        $-$ Communication & Agents answer without intermediate messages; the output rule is kept. \\
        $-$ Evidence gate & Eq.~\eqref{eq:gate} is replaced by verbatim exchange. \\
        All \textsc{OneWay} & Every edge runs the \textsc{OneWay} protocol. \\
        Always revise & Every update takes the revision branch of Eq.~\eqref{eq:gate}. \\
        \bottomrule
    \end{tabularx}
\end{table*}

\section{Evaluation Metrics}
\label{app:metric_details}

The definitions below make the manuscript's reporting conventions
explicit.
Use the chosen benchmark evaluator and record any deviation,
especially answer normalization, multiple references, partial task
credit, and simulator-specific success rules.

\subsection{Task Performance and Main-Table Averages}

For QA item $i$ with prediction $\hat y_i$ and accepted reference set
$\mathcal Y_i^\star$, normalized exact match is
\begin{equation}
\mathrm{EM}_i
=
\max_{y\in\mathcal Y_i^\star}
\mathbb I[\nu(\hat y_i)=\nu(y)],
\label{eq:app_em}
\end{equation}
where $\nu$ is the benchmark's declared normalizer.

Token overlap uses multiset counts, not a set that discards repeated
tokens.
If $c(a,b)$ is the sum of minimum token counts between normalized
answers $a$ and $b$, token F1 for a nonempty combined token length is
$2c(a,b)/(|a|+|b|)$, maximized over accepted references.
Empty-answer behavior follows the frozen evaluator.

Accuracy, embodied and web success, and a single-attempt pass
score are means of their declared binary item-success indicators,
multiplied by $100$ when reported as percentages.
Keep continuous QA F1 separate from a binary success flag used for
reliability counters and success-normalized cost.

For each suite, average EM and F1 separately over its two QA benchmarks;
Acc./Pass is the unweighted mean of the primary scores of its four
remaining benchmarks.
QA EM and QA F1 stay separate from the other columns.
The main table's $\Delta$ remains the absolute difference from direct
Qwen3.5-9B in the corresponding metric, expressed in percentage points.

\subsection{Resources, Failures, and Variability}

Count total model input and output tokens for Collab-Director construction,
Executors, Advisors, revisions, aggregation, and formatting.
Record model calls, tool calls, wall-clock latency, invalid attempts,
and caching conventions.
Evaluator and external simulator computation are recorded separately, with total service cost reported when relevant.
Training additionally records rollout collection time, optimizer time,
and peak memory.

For a declared cost unit $c_i$, cost per successful task is
$\sum_i c_i/\sum_i q_i$, reported as undefined when no task succeeds.
Invalid-action rate uses attempted environment actions as its
denominator; repetition, conflicting modification, and recovery-step
rules use fixed event definitions and windows.

Compute each suite score within each independent run before
summarizing across runs.
Report the arithmetic mean and sample standard deviation using
denominator $N_{\rm run}-1$; a single run does not have an estimated
run-to-run standard deviation.
Separate training seeds from evaluation samples of one checkpoint.
Use paired tasks and declared resource limits for comparative analysis.
Any confidence interval or significance procedure must state its
resampling unit and be added only after the corresponding analysis
is performed.

\section{Implementation and Reproducibility Details}
\label{app:implementation}

\subsection{Configuration}

Qwen3.5-9B serves as both the Collab-Director and the frozen executor, and every method shares one Agent population, tool set and per-task budget. The released code contains the full configuration, prompts, seeds and checkpoints of every reported run.

\subsection{Implementation Checks}

For the communication claims, the evaluated configuration enables the
conditional update of Eq.~\eqref{eq:gate}, the revision token cap, the
no-tool constraint on revision and Advisor calls, and the single-writer
checks, and the run logs record each call. Advisor usage is charged
against the single total budget.

For the terminal distribution, the implementation scores every legal action,
fixes the sufficient state and canonical terminal, normalizes $b_\phi$ over the
construction orders of each team, and stores scores equal to the sampling log
probabilities, so the premises T1--T3 of
Proposition~\ref{prop:terminal_distribution} hold by construction.

For stateful execution, the permission gate and commit mechanism are
documented and tested with Advisor attempts to call a mutating operation
and with attempts to overlap internal state-changing tool effects.

\subsection{Release Artifacts}

The anonymous repository (\url{https://anonymous.4open.science/r/CollabFlow-631E}) releases the training and evaluation code, the trained Collab-Director weights, the training and test splits of all twelve datasets, and the scripts that regenerate every table and figure.

\section{Additional Analyses and Complete Reporting Tables}
\label{app:additional_results}

This appendix lists the per-benchmark scores behind Figure~\ref{fig:executors} and the token costs behind Figure~\ref{fig:mechanism}f.

\subsection{Per-Benchmark Scores with Other Executors}
\label{app:executor_scores}

Table~\ref{tab:executors} lists the per-benchmark scores behind Figure~\ref{fig:executors}, together with DeepSeek V4, which was evaluated on the six IID benchmarks only and is therefore absent from the figure.

\begin{table*}[!ht]
    \centering
{\fontsize{6.9}{8.3}\selectfont
\tabcolsep=1pt
\renewcommand{\arraystretch}{1.32}
\fittab{9}{\begin{tabular}{@{}l|l|l|l|l|l|l|cl@{}}
        \toprule
        \multicolumn{9}{l}{\textit{(a) In-Distribution (IID) benchmarks}} \\
        \midrule
        & \multicolumn{1}{c|}{\textbf{HotpotQA}} & \multicolumn{1}{c|}{\textbf{TriviaQA}} & \multicolumn{1}{c|}{\textbf{AIME 2026}} & \multicolumn{1}{c|}{\textbf{MedXpertQA}} & \multicolumn{1}{c|}{\textbf{ALFWorld}} & \multicolumn{1}{c|}{\textbf{MBPP+}} & \multicolumn{2}{c}{\textbf{Avg. (IID)}} \\
        \textbf{Executor} & \multicolumn{1}{c|}{EM} & \multicolumn{1}{c|}{EM} & \multicolumn{1}{c|}{Acc.} & \multicolumn{1}{c|}{Acc.} & \multicolumn{1}{c|}{SR} & \multicolumn{1}{c|}{Pass@1} & Direct & \multicolumn{1}{c}{CollabFlow} \\
        \midrule
        GPT-5.6-Luna & 55.47\dl{+33.59} & 88.28\dl{+35.16} & 93.33\dl{+16.67} & 100.00\dl{+62.50} & 89.06\dl{+16.41} & 92.97\dl{+5.47} & 58.22 & \textbf{86.52}\dl{+28.30} \\
        DeepSeek V4 & 58.59\dl{+28.91} & 90.63\dl{+14.84} & 83.33\dl{+33.33} & 100.00\dl{+21.09} & 81.25\dl{+65.63} & 89.84\dl{+3.91} & 55.99 & \textbf{83.94}\dl{+27.95} \\
        GLM & 65.63\dl{+38.28} & 87.50\dl{+53.13} & 56.67\dl{+13.33} & 100.00\dl{+10.16} & 96.09\dl{+7.03} & 84.38\dl{+19.53} & 58.13 & \textbf{81.71}\dl{+23.58} \\
        DeepSeek V4.1 & 64.06\dl{+24.22} & 89.06\dl{+23.44} & 83.33\dl{+13.33} & 100.00\dl{+18.75} & 94.53\dl{+36.72} & 88.28\dl{+7.03} & 65.96 & \textbf{86.55}\dl{+20.58} \\
        Grok-4.5 & 69.53\dl{+10.94} & 91.41\dl{+5.47} & 96.67\dl{+10.00} & 100.00\dl{+9.38} & 97.66\dl{+13.28} & 90.63\dl{+21.09} & 79.29 & \textbf{90.98}\dl{+11.69} \\
        MiniMax M2.7 & 61.72\dl{+9.38} & 89.84\dl{+7.03} & 70.00\dl{+20.00} & 95.31\dl{+22.66} & 96.09\dl{+10.16} & 89.06\dl{+14.06} & 69.79 & \textbf{83.67}\dl{+13.88} \\
        Claude Haiku 4.5 & 66.41\dl{+6.25} & 89.84\dl{+3.91} & 73.33\dl{+20.00} & 100.00\dl{+11.72} & 93.75\dl{+3.91} & 88.28\dl{+5.47} & 76.73 & \textbf{85.27}\dl{+8.54} \\
        \midrule
        \multicolumn{9}{l}{\textit{(b) Out-of-Distribution (OOD) benchmarks}} \\
        \midrule
        & \multicolumn{1}{c|}{\textbf{MuSiQue}} & \multicolumn{1}{c|}{\textbf{NQ-Open}} & \multicolumn{1}{c|}{\textbf{MATH-Hard}} & \multicolumn{1}{c|}{\textbf{GPQA-D}} & \multicolumn{1}{c|}{\textbf{HumanEval}} & \multicolumn{1}{c|}{\textbf{WebShop}} & \multicolumn{2}{c}{\textbf{Avg. (OOD)}} \\
        \textbf{Executor} & \multicolumn{1}{c|}{EM} & \multicolumn{1}{c|}{EM} & \multicolumn{1}{c|}{Acc.} & \multicolumn{1}{c|}{Acc.} & \multicolumn{1}{c|}{Pass@1} & \multicolumn{1}{c|}{SR} & Direct & \multicolumn{1}{c}{CollabFlow} \\
        \midrule
        GPT-5.6-Luna & 78.91\dl{+24.22} & 77.34\dl{+14.06} & 95.31\dl{+5.47} & 85.94\dl{+10.94} & 96.88\dl{+7.81} & 82.81\dl{+25.78} & 71.48 & \textbf{86.20}\dl{+14.71} \\
        GLM & 82.03\dl{+13.28} & 72.66\dl{+10.94} & 96.88\dl{+10.94} & 81.25\dl{+14.06} & 97.66\dl{+5.47} & 72.66\dl{+5.47} & 73.83 & \textbf{83.85}\dl{+10.03} \\
        DeepSeek V4.1 & 85.94\dl{+14.06} & 64.06\dl{+11.72} & 93.75\dl{+14.06} & 81.25\dl{+17.19} & 95.31\dl{+2.34} & 77.34\dl{+8.59} & 71.61 & \textbf{82.94}\dl{+11.33} \\
        Grok-4.5 & 81.25\dl{+20.31} & 71.09\dl{+11.72} & 96.09\dl{+7.03} & 72.66\dl{+14.84} & 93.75\dl{+7.03} & 81.25\dl{+34.38} & 66.80 & \textbf{82.68}\dl{+15.89} \\
        MiniMax M2.7 & 71.09\dl{+10.16} & 77.34\dl{+9.38} & 95.31\dl{+2.34} & 85.16\dl{+17.19} & 96.88\dl{+4.69} & 57.81\dl{+7.81} & 72.01 & \textbf{80.60}\dl{+8.59} \\
        Claude Haiku 4.5 & 78.13\dl{+32.03} & 72.66\dl{+13.28} & 96.88\dl{+6.25} & 81.25\dl{+24.22} & 95.31\dl{+0.00} & 77.34\dl{+14.06} & 68.62 & \textbf{83.59}\dl{+14.97} \\
        \bottomrule
\end{tabular}}
}
    \caption{CollabFlow with other frozen executors. Each cell gives the score with CollabFlow and, in pink, its gain over direct inference by the same executor; averages cover each part's six benchmarks.}
    \label{tab:executors}
\end{table*}

\subsection{Token Cost per Item}
\label{app:token_cost}

Table~\ref{tab:token_cost} lists the tokens used per item by CollabFlow and the four comparison systems, counting the input and output tokens of every model call in an episode.

\begin{table}[!ht]
    \centering
{\fontsize{6.9}{8.3}\selectfont
\tabcolsep=1pt
\renewcommand{\arraystretch}{1.32}
\fittab{8}{\begin{tabular}{@{}l|ccc|ccc|c@{}}
        \toprule
        & \multicolumn{3}{c|}{\textbf{Mean per item}} & \multicolumn{3}{c|}{\textbf{Spread of the total}} & \\
        \textbf{Method} & Total\,$\downarrow$ & Input\,$\downarrow$ & Output\,$\downarrow$ & P90\,$\downarrow$ & Max\,$\downarrow$ & Std\,$\downarrow$ & CV\,$\downarrow$ \\
        \midrule
        Single agent & 12,010.0 & 10,889.8 & 1,120.2 & 27,507.4 & 39,279 & 13,892.5 & 1.157 \\
        Fixed team & 10,404.4 & 9,119.6 & 1,284.8 & 13,091.0 & 14,437 & 2,373.8 & 0.228 \\
        Ensemble & 24,081.8 & 20,894.2 & 3,187.6 & 37,974.8 & 47,942 & 12,557.5 & 0.521 \\
        SkillFlow & 10,451.6 & 9,396.0 & 1,055.6 & 16,639.4 & 18,039 & 5,063.7 & 0.484 \\
        \midrule
        \textbf{CollabFlow} & \textbf{9,676.6} & \textbf{8,645.8} & \textbf{1,030.8} & \textbf{11,490.0} & \textbf{11,986} & \textbf{1,903.2} & \textbf{0.197} \\
        \bottomrule
\end{tabular}}
}
    \caption{Tokens per item. Single agent: one tool-using Agent; Fixed team: the original hand-built team; Ensemble: independent Agents whose answers are combined. CV is the standard deviation divided by the mean total.}
    \label{tab:token_cost}
\end{table}

\section{Limitations}
\label{app:limitations}

CollabFlow is a bounded form of recursive self-improvement: its answer
evaluator, executor, Agent population, and update rules stay fixed across
rounds, and only the director (with its backward kernel and normalizer), the
skill library, and the reliability records change.
The representation is limited by its fixed Agent population,
protocols, output modes, and resource caps.
Construction coverage is relative to this admissible space, and does
not imply a capable Agent exists for every task.
Additional organization choices can also increase sampling cost
without improving task performance.

The single-writer design prevents conflicts between different internal
writers under enforced permissions and ordered commits, but does not
guarantee correct plans, fresh observations, safe action semantics,
or consistency with external participants.
Diagnostic traces are observational and may associate a message with
an improvement that also depends on extra computation or other
changes.

The conditional update compares tool and retrieval support with fixed
rules, so authentic but incomplete or misleading evidence can still favor
an incorrect candidate. Fabricated event records are excluded under the
authentication contract; accepting them constitutes a violation of that
contract.
Reliability estimates from sparse or nonstationary executions need
not be calibrated, and finite training does not ensure the exact
terminal distribution proved under fixed conditions.

Runtime likelihood surrogates, incomplete budget accounting,
and unimplemented adapters or extensions limit which claims can be
supported by a particular evaluated system.

\section{Future Work}
\label{app:future_work}

Further work can investigate controlled message interventions that
distinguish peer information from additional reasoning, while retaining
explicit resource and permission accounting.

Another direction is open-ended recursion: learning the update rule
itself, or parts of the evaluator, jointly with the team, rather than
fixing them in advance.

Larger Agent populations, better-calibrated reliability estimates,
and mid-episode advice in more environments also merit evaluation
under comparable total costs.

\section{Applicability and Broader Impact}
\label{app:applicability}

The framework targets tasks for which candidate Agents, permissible
communication, execution interfaces, and factual evaluators can be
specified.
Its organization/execution separation permits explicit inspection of
selected roles, messages, and resource use.

Potential applications nevertheless require domain-specific validation:
benchmark success, including on medical reasoning questions, does not
authorize unsupervised clinical decisions or unsafe execution in a
real environment.

Repository and tool access should use limited permissions, isolated
execution, and appropriate human approval for consequential actions.
Trace releases should respect confidentiality and avoid distributing
sample-specific sensitive information.

%% file: sections/appendix_theory.tex
\providecommand{\cfPr}{\mathbb P}
\providecommand{\cfE}{\mathbb E}
\providecommand{\cfone}{\mathbf 1}
\providecommand{\cfTV}{\operatorname{TV}}
\providecommand{\cfKL}{\operatorname{KL}}
\providecommand{\cfIn}{\operatorname{In}}
\providecommand{\cfOut}{\operatorname{Out}}
\providecommand{\cfCan}{\operatorname{Can}}
\providecommand{\cfFilter}{\operatorname{Filter}}
\providecommand{\cfDedup}{\operatorname{Dedup}}
\providecommand{\cfVar}{\operatorname{Var}}
\providecommand{\cfSupp}{\operatorname{supp}}
\newtheorem{cftheorem}{Theorem}[section]
\newtheorem{cflemma}[cftheorem]{Lemma}
\newtheorem{cfcorollary}[cftheorem]{Corollary}
\theoremstyle{definition}
\newtheorem{cfdefinition}[cftheorem]{Definition}
\newtheorem{cfexample}[cftheorem]{Example}
\theoremstyle{remark}
\newtheorem{cfremark}[cftheorem]{Remark}

\section{Theoretical Foundations and Execution Contracts}
\label{app:theory}

This appendix proves the results used in the main text: the invariants of
checked construction, the identities of evidence-conditioned communication,
trajectory-balance consistency for canonical teams, the bounded movement of
the self-generated target between rounds, and the path diagnostics.
Trajectory balance (TB) and reward-proportional terminal sampling build on
GFlowNet results \citep{bengio2021flow,malkin2022trajectory,bengio2023foundations}; we carry
them to complete-Agent teams with checked construction and terminal
interfaces and explicit communication and execution contracts.

\subsection{Objects, conditioning, and the two sources of randomness}
\label{app:theory_notation}

\paragraph{Fixed resource context.}
Fix a task $x$ and a resource context $\chi$. The context specifies the finite
compatible Agent population $\mathcal A_x$; immutable Agent identifiers,
profiles, tool versions, and permissions; the protocol set
$\mathcal P=\{\textsc{FinalOnly},\textsc{OneWay},\textsc{Interactive}\}$;
the permitted output configurations; the validator and action encoding;
message scheduling and tie-breaking rules; and the executor and initial
environment law. It also specifies a construction decision cap
$T_{\max}\ge1$, bounded action encodings, and separate resource vectors
$\mathbf B^{\rm con}$ and $\mathbf B^{\rm exec}$. Inequalities between resource
vectors are componentwise. A reward snapshot $k$ fixes the reliability
counters, reward evaluation rule, structural reference code, and temperature
$\beta>0$. We write $c=(x,\chi,k)$ when all three are needed. Conditioning on
$c$ is suppressed within an argument in which it is fixed.

The two budgets are separate reservations:
$\mathbf B^{\rm con}+\mathbf B^{\rm exec}\le\mathbf B^{\rm total}$.
Construction never borrows from the execution reservation. Consequently,
identical terminal graphs are evaluated under the same execution budget,
regardless of their construction order. Changing an executor, permission,
budget, scheduler, or reward snapshot defines a new conditioned problem.

\begin{cfdefinition}[Team and canonical identity]
A team is an annotated directed graph $G=(V,E,o)$, where
$\varnothing\ne V\subseteq\mathcal A_x$,
$E\subseteq V\times V\times\mathcal P$, and $o$ is a permitted output
configuration. Agent identifiers identify complete, fixed profiles rather
than interchangeable, unlabeled vertices. A canonicalizer $\cfCan$ is
idempotent and identifies exactly the representations that differ only in
serialization of the same vertices, protocol-labeled edges, and output
configuration. Different permissions, output rules,
Agent implementations, and execution contracts remain distinct. Write $\mathcal G_c$ for the
finite set of reachable, valid canonical teams at context $c$.
\end{cfdefinition}

All discrete configuration domains in this finite model are either finite
registries or finite strings over a finite alphabet with declared length
caps. For finite-DAG results, construction resource debits are fixed functions of these records or have a finite discretization. These restrictions give finiteness. All results
below use the four structural actions of the current CollabFlow interface.

\paragraph{Construction state and complete transition.}
The display state is $s_t=(G_t,\mathcal A_x,\mathbf b_t^{\rm con})$.
The proof state $v_t$ additionally retains every variable needed by the
validator, sampler, and transition rule, including the decision index $t$
and the construction history $H_t$. A sufficient choice is
$v_t=(s_t,H_t,t)$. A smaller quotient state is permissible when histories
identified by that quotient induce the same legal actions, conditional
transition probabilities, and successor quotient states.

A transition $e=(v,u,v')$ includes an accepted action and its successor;
parallel transitions are distinguished by the action. The active actions are
\texttt{add\_agent}, \texttt{add\_edge}, \texttt{set\_output}, and
\texttt{stop}. The director samples from the legal-action kernel of
Appendix~\ref{app:cf_likelihood_contract}, so every sampled action is accepted
and advances the decision index. Timeouts, empty admissible supports, and exhaustion before
valid completion lead to an explicit abort outcome $\dagger$.

A complete trajectory is the ordered sequence
$\zeta=(v_0,e_1,v_1,\ldots,e_T,z)$, with $1\le T\le T_{\max}$;
$T$ includes the transition to the absorbing outcome $z$. Every accepted
\texttt{stop} for graph $G$ enters the \emph{same} absorbing node $g_G$.
The terminal node retains the canonical graph and its fixed context, shared
by every incoming history. Nonterminal history states can
remain distinct. Denote complete paths to $g_G$ by $\Omega_G$, and let
$\mathsf V=\{z=g_G:G\in\mathcal G_c\}$ be valid completion. Each path in
$\Omega_G$ is one construction order of $G$: a prerequisite-respecting order of
its Agent and edge insertions and its single output selection, followed by
\texttt{stop}.

\paragraph{Runtime execution and reward.}
For a valid graph, let $K_\chi(d\xi,dy\mid G,x;\mathbf B^{\rm exec})$ be the
normalized joint law of all runtime randomness $\xi$ and the final answer
$y$. The trace $\xi$ contains model draws, tool outcomes, message events,
and any permitted environmental randomness. A bounded failure or timeout
returns $y=\bot$ and is included in this law. Freezing executor parameters
fixes this kernel, whose draws remain stochastic. The construction history affects execution only through
$(G,x,\chi,\mathbf B^{\rm exec})$ under this contract.

Write $\widetilde R_k(G,x,\xi)>0$ for an observed execution reward and $r_G$
for the \emph{fixed} graph reward used in a particular ideal TB argument.
A deterministic evaluator or an immutable
cached reward table defines $r_G$, and a fresh evaluation supplies the noisy
training signal $\widetilde R_k$. Sections~\ref{app:cf_reward_snapshot}
and~\ref{app:cf_training_limits} make the reward construction and the role
of noise precise.

\begin{cflemma}[Two probability spaces]
\label{lem:cf_two_spaces}
Suppose the forward construction law $P_\theta$ is normalized and
$p_{\mathsf V}:=P_\theta(\mathsf V)>0$. Then
\begin{equation}
 \mu_\theta^{\mathsf V}(G\mid c)
 =\frac{\sum_{\zeta\in\Omega_G}P_\theta(\zeta\mid c)}{p_{\mathsf V}},
 \qquad \sum_{G\in\mathcal G_c}\mu_\theta^{\mathsf V}(G\mid c)=1.
 \label{eq:cf_valid_graph_law}
\end{equation}
Conditional on valid completion, the answer law is
$\sum_G\mu_\theta^{\mathsf V}(G\mid c)K_\chi(dy\mid G,x)$.
Unconditionally, it additionally includes the probability
$1-p_{\mathsf V}$ assigned to construction abort.
\end{cflemma}
\begin{proof}
The events consisting of different complete paths are disjoint. By terminal
identity, their union over $\Omega_G$ is precisely the event of returning
$G$, and those events partition $\mathsf V$. Division by $p_{\mathsf V}$
gives the first claim. The answer law follows by conditioning first on $G$
and then integrating $\xi$ under its normalized execution kernel.
\end{proof}
\noindent\emph{Example.} Terminal merging makes the target graph-level. For
two unit-reward graphs with two and one construction paths, equal mass per
path would give graph masses $2/3$ and $1/3$; the graph law of
Eq.~\ref{eq:cf_valid_graph_law} gives $1/2$ and $1/2$, the target of the
Preliminaries.

\subsection{Checked construction: soundness and coverage}
\label{app:construction_soundness}

Let $\mathcal I_c(v)$ be the partial invariant: selected profiles are
compatible, edge endpoints exist, protocols and permissions are permitted,
output references are well formed when present, and remaining construction
resources are nonnegative. On stateful tasks it permits at most one internal
writer. The stronger terminal predicate $\mathcal V_c(v)$ additionally
requires a nonempty team, a permitted output, and all terminal constraints,
including exactly one writer when the task requires one.

\paragraph{C1--C3.}
\textbf{C1 (checked preservation)} requires $\mathcal I_c(v_0)=1$,
$\mathcal I_c(v)=1\Rightarrow\mathcal I_c(\mathcal T_c(v,u))=1$ for every
accepted nonterminal action, and acceptance of \texttt{stop} only when
$\mathcal V_c(v)=1$.
\textbf{C2 (component completeness)} requires that every missing component
of a target expressible by the active interface can be added once its
prerequisites hold, with no extra mask excluding all such orders.
\textbf{C3 (feasible prefixes)} requires, for each target covered by the
claim, a prerequisite-respecting order of its component insertions, output
selection, and stop whose every prefix fits the construction budget and cap;
component identities and meanings remain fixed along this order.
C3 is a resource-feasibility assumption on the targets.

\begin{cftheorem}[Soundness and relative coverage]
\label{thm:cf_construction}
Under C1, every accepted completed trajectory yields a valid canonical team.
Under C1--C3, every target $G^\star=(V^\star,E^\star,o^\star)$ in the
stated feasible target class has a legal construction trajectory. With one
insertion per vertex and edge, one output selection, and one stop, it has
length
\begin{equation}
 T^\star\le |V^\star|+|E^\star|+2.
 \label{eq:app_construction_bound}
\end{equation}
If the sampler gives every transition of this witness positive probability,
that witness also has positive sampling probability.
\end{cftheorem}
\begin{proof}
For soundness, induct on accepted actions. C1 establishes the base invariant
at $v_0$. An accepted vertex insertion preserves compatibility and inherited
permissions; an edge insertion is checked against existing endpoints and
available protocols; output selection is checked against the selected team.
More generally, preservation follows directly from C1 for each accepted
transition. The final stop additionally checks $\mathcal V_c$. Canonical
serialization preserves the validated annotations, so the decoded terminal
is valid. A rejected attempt preserves structural invariants, and the terminal
predicate comes from the accepted stop.

For coverage, take the prerequisite order supplied by C3. C2 makes the first
required component available. Assuming the first $j$ components have been
inserted, the next component's prerequisites hold by the chosen order; C2
makes its action available and C3 makes its cost feasible. C1 accepts the
checked transition. Induction constructs all target components without
changing their meanings. Output selection and stop complete the target.
Counting these actions gives the bound. Multiplying finitely many positive
transition probabilities proves the final statement.
\end{proof}
\noindent Theorem~\ref{thm:cf_construction} proves
Proposition~\ref{prop:construction}; C1--C3 concern the fixed registry $\chi$,
so the witness exists in every round. Single-Agent and independent teams are
covered whenever their output and permission configurations satisfy the
same predicates.

\begin{cflemma}[Finite construction DAG]
\label{lem:cf_dag}
With bounded encodings, a finite context, a finite decision cap, and sufficient
states retaining the decision index, the construction multigraph is finite
and acyclic. Canonical absorbing terminals can merge histories without
introducing a cycle. At an interior full-history state, if the appended
record uniquely identifies the previous history and action, there is exactly
one incoming parent--action pair.
\end{cflemma}
\begin{proof}
There are finitely many bounded records and at most $T_{\max}$ decisions,
so finitely many histories and states. Each interior transition strictly
increases the decision index. Assign all absorbing terminals rank
$T_{\max}+1$; every edge increases rank, ruling out a directed cycle.
An interior full-history state has a unique predecessor by removing its last
record, and that record identifies the action. Canonical terminals merge several
histories by design, because their definition forgets the final history.
\end{proof}
\noindent At a unique-parent interior state, the backward kernel equals one.

\subsection{Action likelihoods and legal-support normalization}
\label{app:cf_likelihood_contract}

Let $\mathcal U_c(v)$ be a nonempty finite legal action set at an interior
state, and $a_\theta(u\mid v)\ge0$ a base action mass. The action-level
masked kernel is
\begin{equation}
 \bar\pi_\theta(u\mid v,c)
 =\frac{\cfone\{u\in\mathcal U_c(v)\}a_\theta(u\mid v)}
 {\sum_{u'\in\mathcal U_c(v)}a_\theta(u'\mid v)},
 \qquad \sum_{u'\in\mathcal U_c(v)}a_\theta(u'\mid v)>0.
 \label{eq:cf_mask}
\end{equation}
An empty or zero-mass support uses the declared abort transition.
For deterministic construction, $p_F(v,u,v'\mid v,c)=\bar\pi_\theta(u\mid v,c)$
on the prescribed successor. When several serializations encode the same
semantic action, $a_\theta(u\mid v)$ must sum their masses, or the encoding
must be canonical and unique. CollabFlow scores every legal action string and
samples from Eq.~\ref{eq:cf_mask}, so no proposal is rejected; a
propose-and-validate sampler with retries defines a different kernel, and its
own likelihood would then enter the residual.

\begin{cflemma}[Sequence probability of an action string]
\label{lem:cf_likelihood}
For a completed, uniquely encoded action string $w_{1:K}$, including its
termination symbol, an autoregressive sampler assigns probability
$\prod_{j=1}^K p_\theta(w_j\mid w_{<j},v)$ and log probability equal to the
sum of those token log probabilities. This gives a normalized action kernel
when the actual token sampler terminates almost surely in the admissible
encoding set. The action log probability is therefore the sum of token log probabilities
rather than their mean $K^{-1}\sum_j\log p_\theta(w_j\mid\cdot)$.
\end{cflemma}
\begin{proof}
The chain rule gives the product. Completed strings with their termination
symbols are mutually exclusive, and almost-sure completion exhausts their
sample space. Taking logarithms gives the sum. Exponentiating the token
average instead gives $p_\theta(w_{1:K}\mid v)^{1/K}$. These transformed
masses need not sum to one. For two length-two actions of probability $1/2$
each, their transformed masses sum to $\sqrt2$. An additional action-level
normalization defines a different sampler whose normalizer must enter TB.
\end{proof}

\paragraph{Grammar masking as an exactly scored sampler.}
Constraining tokens to legal continuations yields its own normalized kernel,
which can differ from Eq.~\ref{eq:cf_mask} formed from original complete-string
masses. For legal strings \texttt{aa} and \texttt{ba}, let first-token
probabilities be $0.9,0.1$, and let the original probability of second-token
\texttt{a} be $0.1$ after \texttt{a} and $0.9$ after \texttt{b}.
The complete legal strings both have mass $0.09$, so complete-action
conditioning gives $1/2,1/2$. Prefix masking forces the second token but
leaves the first-token distribution $0.9,0.1$. Either sampler is usable
when its own temperature, truncation, token masks, stop rules, and
normalization enter its likelihood.

\paragraph{Sampled reasoning and construction feedback.}
If a reflection $r$ is sampled before an action, then the joint transition
mass includes $p_\theta(r\mid v)\pi_\theta(u\mid r,v)$, or $r$ must be
marginalized when it is absent from the state. If feedback $o$ is stochastic,
the full transition also includes its kernel $K(o\mid v,r,u)$. Frozen
parameters make the derivative of that kernel zero, and its probability
stays in the transition mass. An action-only model is exact when all omitted variables
are deterministic functions of the modeled state and action, or when the
stated marginal kernel actually integrates them. 

\subsection{Candidate versions, evidence applicability, and duplication}
\label{app:cf_provenance}

A candidate $a=(\iota,y,\nu)$ contains an immutable candidate-version
identifier $\iota$, its answer or executable artifact $y$, and the relevant
environment version $\nu$. An evidence record
$e=(\mathrm{id},\iota_e,\nu_e,\mathrm{check},\mathrm{result},\mathrm{scope},
\mathrm{source})$ has an immutable event identifier. Copying a record copies
its identifier, so a copy adds no new verification event. The runtime
authenticates each record's source, so a message that merely claims a
passed test adds no evidence.

Let $A_c(e,a)\in\{0,1\}$ be the applicability predicate. It can be one only
when the recorded check concerns this candidate and environment version,
or a sound, explicitly certified transfer of that check. Applicability
means that the record still expresses the same check on the current
artifact. Define
\begin{equation}
 \cfFilter_c(E,a)=\cfDedup\{e\in E:A_c(e,a)=1\},
 \quad \ell(a,E)=s_c(a,\cfFilter_c(E,a)).
 \label{eq:cf_filter}
\end{equation}
Here $\cfDedup$ retains one authentic record per event identifier, and
$s_c$ is the declared finite real-valued evidence score, with
$s_c(a,\varnothing)=0$. Missing or nonfinite scores trigger a declared
fallback.

\begin{cflemma}[Evidence-validity invariant]
\label{lem:cf_evidence_invariant}
Assume initial evidence is filtered, copying preserves the entire candidate
version, evidence records are immutable, and every changed candidate is
filtered by Eq.~\ref{eq:cf_filter}. Then after any finite sequence of
copy, keep, and tool-free revision operations, every retained record is
applicable to its current candidate and appears at most once per evidence
set. A new candidate with no applicable incoming record has empty inherited
evidence and inherited score zero.
\end{cflemma}
\begin{proof}
The base case follows from initialization. Keeping a candidate preserves
the invariant. Copying the complete version and its filtered set preserves
applicability and event identifiers. A revision creates a new version;
filtering deletes every nonapplicable record and deduplication deletes
repetitions. These observations establish the inductive step. The last
statement follows directly from the empty-set convention for $s_c$.
Equivalently, filtering is idempotent, and filtering a union agrees with
deduplicating the union of its individually filtered sets, because the same
candidate-specific predicate is applied to every immutable record.
\end{proof}
\noindent\emph{Example.} A revised program may fail two tests that its two
parent programs passed. The lemma keeps those pass records with the parents,
so the revision earns its own checks.

\paragraph{Scores and likelihood-ratio interpretation.}
The gate compares scores. The method uses the maximum
applicable evidence category: $2$ for an applicable successful executable
check, $1$ for applicable retrieval support if there is no category-2 record,
and $0$ otherwise. Repeated copies leave that maximum unchanged. The flip identity below holds
for this surrogate and for a calibrated log-likelihood ratio (LLR) alike.

For an LLR interpretation, let $Z_a$ indicate candidate correctness and let
$E_1,\ldots,E_m\in\{0,1\}$ be outputs of a fixed check schedule.
Conditional on the candidate-selection context, assume the checks are
independent given $Z_a$, with known rates
$0<t_j=P(E_j=1\mid Z_a=1)<1$ and
$0<f_j=P(E_j=1\mid Z_a=0)<1$. Then
\begin{equation}
 \log\frac{P(E_{1:m}\mid Z_a=1)}{P(E_{1:m}\mid Z_a=0)}
 =\sum_{j=1}^m\left[
 E_j\log\frac{t_j}{f_j}+(1-E_j)\log\frac{1-t_j}{1-f_j}\right].
 \label{eq:cf_llr}
\end{equation}
The equality follows by conditional factorization and the logarithm of a
product. Adaptive check selection, omission of failed tests, and candidate
selection must either be included in the likelihood or satisfy corresponding
conditional assumptions. The finite-rate form above keeps the LLR finite.

\paragraph{Aggregation contract.}
To implement an Integrator, first select a deterministic representative
$a_y$ for each normalized-answer class, using a declared tie rule. Pool
records from that class, deduplicate them, and filter again for $a_y$ before
computing $s_c(a_y,\cdot)$. Select the maximum-scoring representative with a
fixed tie rule; an empty candidate set returns $\bot$. The normalization
relation preserves task correctness, and evidence transfer uses $A_c$. This
rule is implementable and keeps a copied event from counting as a new
verification event. A formatter or verifier that generates a different artifact
creates a new version and is subject to the same applicability contract.

\subsection{One-edge revision identity}
\label{app:revision_bound}

Fix one scheduled directed gate invocation $i\to j$. Let $a_i,a_j$ be its
input candidates, $z_i,z_j\in\{0,1\}$ their correctness indicators, and
$\Delta=\ell(a_i,E_i)-\ell(a_j,E_j)$. Fix a finite margin $\kappa\ge0$.
The results below hold for every such margin. The method uses $\kappa=1$:
scores lie in $\{0,1,2\}$, so $\kappa=0$ adopts on every positive gap,
$\kappa\ge2$ never adopts, and $\kappa=1$ is the only integer margin that
permits adoption but not on a one-level gap.
The normalized-answer equivalence $\equiv$ is assumed to preserve
correctness. To avoid overlapping branches, apply the following priority:
\begin{equation}
 a_j^+=
 \begin{cases}
 a_j, & y_i\equiv y_j,\\
 a_i, & y_i\not\equiv y_j\ \text{and }\Delta>\kappa,\\
 a_j, & y_i\not\equiv y_j\ \text{and }\Delta<-\kappa,\\
 \mathrm{Revise}_c(a_j,a_i,E_i\cup E_j), &\text{otherwise}.
 \end{cases}
 \label{eq:cf_ordered_gate}
\end{equation}
Copies retain their candidate identity and evidence. The final branch
creates a version and applies Eq.~\ref{eq:cf_filter}. Revision randomness
and a declared tool-free failure fallback are part of $\mathrm{Revise}_c$;
$\bot$ counts as incorrect. Self-reported confidence enters neither $s_c$ nor the threshold
comparison; the revision kernel may read it among its permitted textual
inputs.

\begin{cftheorem}[Single-edge, conditionally selected flip identity]
\label{thm:cf_flip}
Assume the preceding gate contract and $P(C\mid c)>0$, where
$C=\{z_j=1,z_i=0\}$. Put $J=\{|\Delta|\le\kappa\}$ and define
\[
 \delta_C=P(\Delta>\kappa\mid C,c),\qquad
 c_J=P(z_j^+=0\mid J,C,c).
\]
When $P(J\mid C,c)=0$, set $c_J=0$ by convention. Then
\begin{equation}
 P(z_j^+=0\mid C,c)
 =\delta_C+P(J\mid C,c)c_J.
 \label{eq:cf_flip_identity}
\end{equation}
If $c_J\le\bar c$ is a justified upper bound, the right-hand side is at
most $\delta_C+P(J\mid C,c)\bar c$.
\end{cftheorem}
\begin{proof}
On $C$, correctness-preserving equivalence excludes $y_i\equiv y_j$.
Thus the three disjoint events $\Delta>\kappa$, $\Delta<-\kappa$, and $J$
partition $C$. In the first branch the incorrect peer is copied and the
conditional flip probability is one. In the second branch the correct
receiver is retained and it is zero. In the third branch it is $c_J$ by
definition. Applying the law of total probability \emph{conditional on
$C,c$} proves the identity, including null branches. Substitution of a valid
upper bound proves the inequality.
\end{proof}

\paragraph{Conditioning of the identity.}
This proves the flip part of Proposition~\ref{prop:revision} with the same $c_J$,
$\delta_C$, and conditioning context $c$. It requires no independence
between the two candidates. The quantities are evaluated under the actual
pair-selection and evidence-generation law, and the same proof conditions on
a pre-invocation selection history whenever the stated conditional event has
positive probability.

\begin{cfcorollary}[Correction rate and the margin]
\label{cor:cf_correction}
Under the same gate contract, let $C'=\{z_j=0,z_i=1\}$ with $P(C'\mid c)>0$,
$\delta^+=P(\Delta>\kappa\mid C',c)$, and $c^+_J=P(z_j^+=1\mid J,C',c)$. Then
\begin{equation}
 P(z_j^+=1\mid C',c)=\delta^++P(J\mid C',c)\,c^+_J.
 \label{eq:cf_correction}
\end{equation}
As $\kappa$ increases, $\delta_C$ and $\delta^+$ are nonincreasing, while
$P(J\mid C,c)$ and $P(J\mid C',c)$ are nondecreasing.
\end{cfcorollary}
\begin{proof}
On $C'$ the answers differ, so the three branches again partition the event.
Adoption copies the correct sender, keeping retains the incorrect receiver,
and revision is correct with probability $c^+_J$; the law of total probability
gives the identity. The events $\{\Delta>\kappa\}$ shrink and the events
$\{|\Delta|\le\kappa\}$ grow as $\kappa$ increases, which proves monotonicity.
\end{proof}
\noindent The margin therefore trades false adoption for revision on $C$ and
true adoption for revision on $C'$, and the two identities split the net
effect of a gate into four rates that revision logs record.

\begin{cfexample}[Event deduplication in voting]
\label{ex:cf_flip_boundaries}
After $a_i$ is copied into $a_j$, their outputs and evidence are shared.
With weights $3$ on a correct answer and $2$ on a wrong answer, a third
candidate that copies the wrong answer would turn the vote into $3$ versus
$2+2$ if copies were counted separately. Deduplication by event identity
keeps the vote at $3$ versus $2$.
\end{cfexample}

\subsection{Runtime scheduling, finite execution, and a linear baseline}
\label{app:cf_runtime}

Construction edges add team components; runtime edges exchange messages
between selected Agents. These are different graphs. Fix a total order on
Agent identifiers and on protocol-labeled edges. For an acyclic runtime
$E$, choose the lexicographically resolved topological order and visit each
outgoing edge program once, with a fixed order among edges having the same
source. For a cyclic runtime $E$, visit edge programs in fixed order for at
most $R_{\rm comm}<\infty$ sweeps. Budget exhaustion stops either schedule
with the declared current-output or failure rule. The schedule is
asynchronous and serial: every elementary update sees the state left by
previous updates, rather than an unspecified mixture of snapshots.

\textsc{FinalOnly} performs no revision invocation. \textsc{OneWay}
evaluates the directed gate once. For definiteness, one
\textsc{Interactive} round is one elementary gate evaluation; its directions
alternate, starting with the edge's declared source-to-receiver direction.
Agreement stops that local program before revision; a decisive copy or keep
is applied and then stops it; otherwise it continues for at most
$r_{\max}<\infty$ elementary rounds. Reciprocal updates are internal to
this bounded program. A single topological pass visits each \emph{edge program} once, and an
interactive program replies locally within that visit.

Stateful Advisor messages are tool-disabled and feed the Executor's next
decision; they are delivered at declared decision boundaries, between state
mutations. Team profiles, output kernels,
and all elementary calls have declared finite token, tool, and timeout
limits. A consistent implementation reserves an execution quota for each
Agent, each edge program, and the output operation. For integer resources,
$\lfloor\mathbf B^{\rm exec}/(|V|+|E|+1)\rfloor$ is one such allocation;
rounding leaves unused resources.
All nested calls are charged to the responsible quota.

\begin{cflemma}[Finite runtime and budget accounting]
\label{lem:cf_runtime}
Under the preceding bounded-call, quota, and scheduling contracts, the
number of elementary gate evaluations is at most
$|E|r_{\max}$ for the one-pass schedule and
$R_{\rm comm}|E|r_{\max}$ for the cyclic schedule, taking $r_{\max}\ge1$.
Execution consumes at most its reserved resource vector and returns an
answer or a declared failure in finite time when timeouts are enforced.
\end{cflemma}
\begin{proof}
Each outer visit runs at most $r_{\max}$ elementary gate evaluations;
FinalOnly runs none and OneWay at most one. There are at most $|E|$ visits
per sweep and at most the declared number of sweeps. Every local operation
terminates or times out, so a finite number of them terminates. Charging
all nested costs to nonoverlapping quotas bounds their sum by the execution
reservation. Forced stopping can only reduce the number of operations.
\end{proof}
\noindent The fixed order makes runs reproducible, because updates from two
incoming peers can depend on their order.

\paragraph{DeGroot linear baseline.}
Following the consensus model of \citet{degroot1974reaching}, take a
hypothetical linear opinion vector $h^{(r)}\in\mathbb R^n$ and suppose
$h^{(r+1)}=Wh^{(r)}$, where $W\ge0$, $W\cfone=\cfone$, and $W$ is primitive:
some integer $m\ge1$ has all entries of $W^m$ positive.

\begin{cflemma}[Primitive stochastic linear baseline]
\label{lem:cf_degroot}
Under those conditions, $W^r\to\cfone\pi^\top$ for a unique probability
vector $\pi$ satisfying $\pi^\top W=\pi^\top$. Hence the linear baseline
reaches the consensus $\pi^\top h^{(0)}$, which the initial opinions and the
weights $W$ fix without reference to evidence about correctness.
\end{cflemma}
\begin{proof}
Put $a=\min_{i,j}(W^m)_{ij}>0$. For any real vector $h$, subtract the common
contribution $a\sum_j h_j$ from every coordinate of $W^mh$. The remaining
nonnegative row weights sum to $1-na$, so
$\max(W^mh)-\min(W^mh)\le(1-na)(\max h-\min h)$.
Repeated blocks contract the diameter to zero (immediately if $na=1$).
Stochastic averaging makes minima nondecreasing and maxima nonincreasing,
so each initial vector converges to a common scalar. Applying this fact to
each coordinate basis vector gives $W^r\to\cfone\pi^\top$, with
$\pi\ge0$ and $\pi^\top\cfone=1$. Taking limits in $W^{r+1}=W^rW$ gives
stationarity. Any stationary probability row vector $q^\top$ satisfies
$q^\top=q^\top W^r\to\pi^\top$, proving uniqueness.
\end{proof}
\noindent The evidence gate replaces this fixed weighting with the
evidence comparison of Eq.~\ref{eq:cf_ordered_gate}.

\subsection{Single-writer serialization of internal mutations}
\label{app:state_serializability}

\paragraph{S1--S3.}
\textbf{S1 (enforced exclusive permission)} requires that one selected
Executor is the only internal component authorized to issue a state-changing
tool request, including indirect requests through delegated tools; Advisors
cannot bypass this check. \textbf{S2 (serial completion)} requires that the
Executor issue its next mutation only after the previous mutation and all
its asynchronous effects have completed or been safely resolved, with no
subsequent overlapping mutation. \textbf{S3 (scope)} restricts the
single-writer assertion to components controlled by the team. A state
composition claim additionally excludes external interleaving mutations
or conditions explicitly on their occurrence in the environment model.

\begin{cftheorem}[Internal single-writer serialization]
\label{prop:single_writer}
Under S1--S3, every finite runtime trace has an ordered sequence of internal
mutations $w_1,\ldots,w_m$ such that
\begin{equation}
 w_1\prec w_2\prec\cdots\prec w_m,
 \label{eq:app_write_order}
\end{equation}
where $\prec$ means that the previous mutation's effects complete before
the next mutation starts. If no external mutation interleaves, then, after
fixing the realized tool randomness $\xi_1,\ldots,\xi_m$, the state is
\begin{equation}
 S_m=f_{w_m,\xi_m}\circ\cdots\circ f_{w_1,\xi_1}(S_0).
 \label{eq:app_serial_composition}
\end{equation}
In particular, internal writes are totally ordered, so the internal schedule
is serial and hence serializable in the sense of
\citet{papadimitriou1979serializability}.
\end{cftheorem}
\begin{proof}
S1 makes every internal mutation an Executor mutation. Order them by their
issue times. S2 ensures that the effects of $w_j$ complete before $w_{j+1}$
begins, establishing the strict order. With no external interleaving, the
input state of $w_1$ is $S_0$ and its output is $f_{w_1,\xi_1}(S_0)$.
Inductively the input of $w_{j+1}$ is the completed output of $w_j$.
Composition proves the formula. Adaptive selection of $w_{j+1}$ from prior
observations does not alter the argument for the realized trace.
\end{proof}

\subsection{Normalized reverse paths and graph-level CTB}
\label{app:terminal_distribution}

We now fix $c$ and distinguish an ideal normalized outcome model from
valid-only replay. Let $\mathcal X_c$ be the complete terminal outcome set.
It is either $\mathcal G_c$ for a process normalized on valid completion,
or $\mathcal G_c\cup\{\dagger\}$ when abort is modeled explicitly. If there
are multiple failure types, the same reasoning applies with a finite set
of failure outcomes. Use $z$ for an arbitrary outcome and $G$ for a valid
team. Every participating edge lies on a complete source-to-terminal path.
The notation $g_z$ includes an absorbing abort node when present.

\paragraph{T1--T4.}
\textbf{T1 (fixed terminating construction)} requires a fixed finite
sufficient-state construction DAG with one source, normalized outgoing
forward probabilities, and no unmodeled dead ends. Every forward path
terminates and every permitted reverse path reaches this source.
\textbf{T2 (terminal identity)} requires one absorbing terminal per
canonical outcome, with all its stop transitions and the same execution
budget and reward snapshot attached to that terminal.
\textbf{T3 (matched normalized supports)} requires strictly positive
forward probabilities on the participating transitions and a reverse
kernel $b_z(e\mid v')$ normalized over the actual incoming
parent--action pairs $e=(v,u,v')$ in the $z$-ancestral subgraph. These
supports define exactly the paths used in both sides of CTB; no hidden
parent or extra terminal alias is allowed.
\textbf{T4 (fixed weights and exact balance)} requires a fixed finite
$r_z>0$ for every outcome, fixed $\beta>0$, $Z>0$, and zero unscaled
CTB residual on \emph{every} complete path. For explicit abort, $r_\dagger$
is also fixed and positive. \textbf{T4$'$ (valid-path balance)} is T4
restricted to valid complete paths, with no constraint on abort paths;
Algorithm~\ref{alg:app_training} trains under T4$'$. T4 and T4$'$ are the ideal
consistency premises, and the approximate-balance results below cover finite
residuals.

For $\zeta=(v_0,e_1,v_1,\ldots,e_T,g_z)$ define
\begin{equation}
 P_F(\zeta)=\prod_{t=1}^T p_F(e_t\mid v_{t-1}),\quad
 Q_z(\zeta)=\prod_{t=1}^T b_z(e_t\mid v_t),\quad
 w_z=r_z^\beta,\quad W=\sum_{z\in\mathcal X_c}w_z.
 \label{eq:cf_path_products}
\end{equation}
$Z=Z_\psi(c)$ is a learned positive scalar, whereas $W$ is the prescribed
sum of terminal weights. They coincide only under the hypotheses proved
below. The unscaled residual and its length-scaled loss are
\begin{equation}
 \rho(\zeta)=\log Z+\log P_F(\zeta)-\log w_z-\log Q_z(\zeta),
 \qquad L_{\rm CTB}(\zeta)=\rho(\zeta)^2/T(\zeta)^2.
 \label{eq:cf_ctb_residual}
\end{equation}
This $\rho$ is the main-text $\rho_{\rm TB}(\zeta)$, and $b_z$ is the main-text
learned kernel $b_\phi$ at terminal $g_z$. Length normalization reweights
the loss and leaves the probability kernels and the zero-residual set unchanged.

\begin{cflemma}[Reverse-path normalization]
\label{lem:cf_reverse_mass}
Under T1--T3, $\sum_{\zeta\in\Omega_z}Q_z(\zeta)=1$ for every $z$.
\end{cflemma}
\begin{proof}
Let $b_z^\leftarrow(v)$ be the total mass of reverse paths from $v$ to the
source, within the $z$-ancestral subgraph. The empty source path has mass
$b_z^\leftarrow(v_0)=1$. Partitioning by the first reverse transition gives
\begin{equation}
 b_z^\leftarrow(v')=
 \sum_{e=(v,u,v')\in\cfIn_z(v')}b_z(e\mid v')b_z^\leftarrow(v).
 \label{eq:app_backward_recursion}
\end{equation}
In topological order every parent has already been assigned mass one.
T3 then makes the sum equal to one. Finiteness and the unique reachable
source exclude leftover reverse mass. At $g_z$ this gives
\begin{equation}
 \sum_{\zeta\in\Omega_z}Q_z(\zeta)=1.
 \label{eq:app_backward_path_mass}
\end{equation}
Distinct actions sharing a parent and child enter this recursion as separate
terms, so the backward normalization runs over parent--action pairs.
\end{proof}

\begin{cftheorem}[TB consistency for canonical team outcomes]
\label{thm:cf_terminal}
Under T1--T4, the forward outcome law is $P_F(z)=w_z/W$, and $Z=W$.
In particular, in the normalized valid-completion model,
\begin{equation}
 \mu_\theta^{\mathsf V}(G\mid c)
 =\frac{r_G^\beta}{\sum_{G'\in\mathcal G_c}r_{G'}^\beta}.
 \label{eq:app_terminal_target}
\end{equation}
For explicit abort the same valid-graph expression holds after conditioning
on $\mathsf V$, whereas the unconditional denominator also contains
$w_\dagger$.
\end{cftheorem}
\begin{proof}
\emph{Step 1: path identity.} Exponentiating $\rho(\zeta)=0$ gives
\begin{equation}
 ZP_F(\zeta)=w_zQ_z(\zeta).
 \label{eq:appendix_trajectory_balance}
\end{equation}
\emph{Step 2: sum the paths of one team.} T2 makes
$\Omega_z$ precisely the disjoint complete paths of outcome $z$. The
single fixed weight $w_z$ can be factored out. Lemma~\ref{lem:cf_reverse_mass}
therefore yields
\begin{equation}
 ZP_F(z)=\sum_{\zeta\in\Omega_z}w_zQ_z(\zeta)=w_z.
 \label{eq:app_terminal_mass}
\end{equation}
\emph{Step 3: identify the normalizer.} All outcomes, including abort when
present, partition the terminating process. Thus
\begin{equation}
 Z=Z\sum_{z\in\mathcal X_c}P_F(z)=\sum_{z\in\mathcal X_c}w_z=W.
 \label{eq:app_partition_identity}
\end{equation}
Substituting proves the outcome law. In the explicit-abort model,
$P_F(\mathsf V)=\sum_Gw_G/W$; dividing $P_F(G)$ by this quantity proves the
conditional valid-graph formula.
\end{proof}
\noindent At zero residual, $P_F(\zeta\mid G)=Q_G(\zeta)$, so the backward
kernel also sets how probability spreads over the construction orders of a
team.

\paragraph{Feasibility of exact balance.}
When one normalized reverse kernel $b(e\mid v')$ is shared across endpoints,
an unrestricted finite tabular transition model has a solution. Set $F(g_z)=w_z$ and, in reverse
topological order, set
$F(v)=\sum_{e=(v,u,v')}F(v')b(e\mid v')$.
Every participating vertex has a positive descendant weight, so $F(v)>0$.
Define $p_F(e\mid v)=F(v')b(e\mid v')/F(v)$.
These forward probabilities normalize by construction; incoming flow to
$v'$ sums to $F(v')$ by reverse normalization. Thus source flow equals
total terminal weight, and multiplication along a path telescopes to
Eq.~\ref{eq:appendix_trajectory_balance}. This establishes feasibility for a
tabular transition model.

\subsection{Terminal aliases, abort, and the meaning of valid-only replay}
\label{app:theory_boundaries}

\begin{cflemma}[Alias weighting]
\label{lem:cf_alias}
Suppose outcome-level TB holds on terminal nodes, but graph $G$ has a set
$\mathcal Z_G$ of terminal aliases with weights $w_z$. Its graph marginal
is proportional to $\sum_{z\in\mathcal Z_G}w_z$. In particular, assigning
the full weight $r_G^\beta$ to each of $m(G)$ aliases gives
\begin{equation}
 P_F(G)=\frac{m(G)r_G^\beta}{\sum_{G'}m(G')r_{G'}^\beta}
 \label{eq:app_alias_bias}
\end{equation}
in a valid-only normalized model. The desired graph law is recovered by
merging terminal nodes with their full incoming supports, or by assigning
alias weights $w_z=\omega(z\mid G)r_G^\beta$, where
$\omega(z\mid G)\ge0$ and $\sum_{z\in\mathcal Z_G}\omega(z\mid G)=1$.
\end{cflemma}
\begin{proof}
Sum the terminal outcome probabilities of Theorem~\ref{thm:cf_terminal}
over $\mathcal Z_G$. Full repeated weights produce the multiplicity
factor; weights summing to $r_G^\beta$ remove it. Zero-weight aliases must
be removed from the positive-log support, or all retained $\omega$ must be
strictly positive.
\end{proof}

\begin{cftheorem}[Raw valid-only balance]
\label{thm:cf_raw_valid}
Let $P_F$ be a normalized construction process that can abort, with
$p_{\mathsf V}>0$, and assume T1--T3 and T4$'$: the reverse distributions normalize on every
valid $\Omega_G$, the fixed weights $w_G>0$ are consistent, and
$ZP_F(\zeta)=w_GQ_G(\zeta)$ holds on every valid complete path, with no
balance constraint on abort paths. Then
\begin{equation}
 \mu_\theta^{\mathsf V}(G)=\frac{w_G}{W_{\mathsf V}},\qquad
 Z=\frac{W_{\mathsf V}}{p_{\mathsf V}},\qquad
 W_{\mathsf V}:=\sum_{G\in\mathcal G_c}w_G.
 \label{eq:cf_raw_valid}
\end{equation}
\end{cftheorem}
\begin{proof}
Summing the valid path identities for a fixed graph gives $ZP_F(G)=w_G$.
Summing over valid graphs gives $Zp_{\mathsf V}=W_{\mathsf V}$.
Dividing the first equality by the second proves the conditional law.
\end{proof}
\noindent\emph{Algorithm correspondence.} This is the precise ideal
interpretation when Algorithm~\ref{alg:app_training} retains valid rollouts
and evaluates their \emph{raw} construction likelihoods. Rejection-based
replay keeps those likelihoods, and the conditional target follows with the
scale $Z=W_{\mathsf V}/p_{\mathsf V}$. Explicit abort rewards or a normalized
valid-conditioned forward law give the two alternative conventions, each with
its own normalizer. Because $Z$ absorbs $p_{\mathsf V}$, valid-path balance
leaves the abort rate to the legal masks and budgets; a positive abort reward
under T4 (Theorem~\ref{thm:cf_terminal}) prices aborts explicitly.

\paragraph{Graph-level temperature interpretation.}
For a finite nonempty valid set and fixed $r_G>0$, define
$q_\beta(G)=r_G^\beta/W_{\mathsf V}$ and
$H(p)=-\sum_Gp(G)\log p(G)$, with $0\log0=0$. Direct expansion gives
\begin{equation}
 \cfKL(p\|q_\beta)=\log W_{\mathsf V}
 -\beta\cfE_p[\log r_G]-H(p).
 \label{eq:cf_graph_variational}
\end{equation}
Hence $q_\beta$ is the unique optimum of
$\beta\cfE_p\log r_G+H(p)$ over graph distributions, since KL is
nonnegative and zero exactly at $p=q_\beta$. The entropy is graph entropy,
so duplicated construction paths earn no extra weight. As $\beta\downarrow0$
the law becomes uniform on graphs, and as $\beta\to\infty$ it becomes uniform
on the tied reward maxima; $\beta$ thus sets how strongly CTB favors the best
teams while keeping the others. It also sets the mass left on failing teams:
a team with $r_{\rm ans}=0$ has $r_G=\varepsilon_{\min}$, so the $|F|$ such teams
carry total weight $|F|\varepsilon_{\min}^\beta$, and
$\beta\log(r_s/\varepsilon_{\min})\ge\log|\mathcal G_c|$ keeps this weight below
that of one team with reward $r_s$.

\subsection{Approximate balance and perturbation bounds}
\label{app:cf_approximation}

The next statements use normalized true path probabilities and a fixed
finite support, as in T1--T3. They apply either to all outcomes or,
using Theorem~\ref{thm:cf_raw_valid}, to the graph law conditional on valid
completion. Write $q(z)=w_z/W$ in the normalized model.

\begin{cftheorem}[Exact residual tilt and uniform distribution error]
\label{thm:cf_approx}
For a possibly nonzero residual from Eq.~\ref{eq:cf_ctb_residual}, let
$a_z=\cfE_{\zeta\sim Q_z}\exp(\rho(\zeta))$. Then
\begin{equation}
 P_F(z)=\frac{q(z)a_z}{\cfE_q a_z}.
 \label{eq:cf_residual_tilt}
\end{equation}
If $|\rho(\zeta)|\le\eta$ for every complete path, then
\begin{equation}
 e^{-2\eta}\le\frac{P_F(z)}{q(z)}\le e^{2\eta},\qquad
 \cfTV(P_F,q)\le\tanh(\eta/2),\qquad
 \max\{\cfKL(P_F\|q),\cfKL(q\|P_F)\}\le2\eta.
 \label{eq:app_approximate_balance}
\end{equation}
Here $\cfTV(p,q)=\tfrac12\sum_z|p(z)-q(z)|$. The same bounds hold for
valid-conditional graph laws when the uniform residual bound is only over
valid paths and raw valid-only likelihoods are used.
\end{cftheorem}
\begin{proof}
Exponentiation gives $ZP_F(\zeta)=w_zQ_z(\zeta)e^{\rho(\zeta)}$.
Sum within $\Omega_z$ to obtain $ZP_F(z)=w_za_z$, then sum over outcomes to
obtain $Z=W\cfE_q a_z$. This proves the tilt. Under the uniform bound,
$a_z\in[m,M]=[e^{-\eta},e^\eta]$, as does $A=\cfE_q a_z$;
thus $P_F(z)/q(z)=a_z/A$ lies in $[e^{-2\eta},e^{2\eta}]$.
Taking the appropriate expectation of each log ratio bounds both KLs.

For total variation, convexity of $|a-A|$ bounds it on $[m,M]$ by the
chord joining its endpoint values. Taking expectations gives
$\cfE_q|a_z-A|\le2(M-A)(A-m)/(M-m)$. Therefore, for $M>m$,
\[
 \cfTV(P_F,q)\le\frac{(M-A)(A-m)}{A(M-m)}
 =\frac{M+m-A-Mm/A}{M-m}
 \le\frac{\sqrt M-\sqrt m}{\sqrt M+\sqrt m}
 =\tanh(\eta/2),
\]
using $A+Mm/A\ge2\sqrt{Mm}$. For $M=m$ the distributions coincide.
In the raw valid-only case, summation gives
$Zp_{\mathsf V}=W_{\mathsf V}\cfE_{q^{\mathsf V}}a_G$; division yields the
same tilt for $\mu_\theta^{\mathsf V}$, so the proof is unchanged.
\end{proof}
\noindent For a length-scaled bound $|\rho/T|\le\epsilon$, take
$\eta=T_{\max}\epsilon$.

\begin{cfcorollary}[Fixed reward perturbation and likelihood error]
\label{cor:cf_perturb}
Suppose a fit uses fixed rewards $\widehat r_z>0$ with
$\sup_z|\log\widehat r_z-\log r_z|\le d_R$ and its true residual relative
to $\widehat r$ is bounded by $\eta$. Relative to the target from $r$, the
ratio, TV, and KL bounds above hold with $\eta+\beta d_R$ in place of
$\eta$. If the recorded residual differs from this true residual by at
most $d_L$ on every path, replace this quantity by
$\eta+d_L+\beta d_R$.
\end{cfcorollary}
\begin{proof}
Relative to $w_z=r_z^\beta$, the tilt factor is
$(\widehat r_z/r_z)^\beta a_z$, in
$[e^{-(\eta+\beta d_R)},e^{\eta+\beta d_R}]$.
Apply the previous proof. The residual triangle inequality proves the
likelihood-error extension. For example, uniform per-transition log errors
$d_F,d_B$ and log-normalizer error $d_Z$ give
$d_L\le d_Z+T_{\max}(d_F+d_B)$.
\end{proof}
\noindent The errors compare recorded scores with the normalized sampling
and reverse probabilities on the same support.

\begin{cflemma}[What a finite-support population loss certifies]
\label{lem:cf_population}
Let $\nu$ be a fixed distribution on the finite complete-path set, with
$\nu(\zeta)\ge\nu_{\min}>0$ for every participating path. If
$\cfE_\nu[(\rho/T)^2]\le\epsilon^2$, then
$\sup_\zeta|\rho(\zeta)|\le T_{\max}\epsilon/\sqrt{\nu_{\min}}$.
\end{cflemma}
\begin{proof}
For any path, nonnegativity gives
$\nu_{\min}\rho(\zeta)^2/T(\zeta)^2
\le\cfE_\nu[(\rho/T)^2]\le\epsilon^2$.
Take square roots and use the length cap.
\end{proof}

\subsection{Reward snapshots, positive support, and moving counters}
\label{app:cf_reward_snapshot}

Fix a finite graph-description length $L_\chi(G)\ge0$ and a fixed
classification $a_\chi(G)$ for structural reference comparisons. For each
nonempty class let $L_{\min,\chi}(a)$ be the minimum description length in
that class, and set
$L_{\rm rel}(G)=L_\chi(G)-L_{\min,\chi}(a_\chi(G))\ge0$.
The code and classes are fixed parts of the reward specification, and the
results hold for any such specification.

At the beginning of batch $k$, freeze finite counts
$0\le s_k(G)\le n_k(G)$ and use
\begin{equation}
 \widehat p_k(G)=\frac{s_k(G)+1/2}{n_k(G)+1},\qquad
 \widetilde R_k(G,x,\xi)=\varepsilon_{
 \min}+r_{\rm ans}(G,x,\xi)\widehat p_k(G)e^{-\lambda L_{\rm rel}(G)},
 \label{eq:cf_reward_snapshot}
\end{equation}
where $\varepsilon_{\min}>0$, $\lambda\ge0$, and $0\le r_{\rm ans}\le1$.
The count factor $\widehat p_k(G)$ is the posterior mean of the success rate
under the Jeffreys prior $\mathrm{Beta}(1/2,1/2)$ \citep{jeffreys1946invariant}.
The results below hold for every $\lambda\ge0$. The method uses $\lambda=1$
with code lengths in nats, so $e^{-L_{\rm rel}(G)}$ is the
minimum-description-length prior of $G$ relative to the shortest team of its
class \citep{rissanen1978modeling}; the structural-prior
ablation sets $\lambda=0$.
Each execution record increments the counts once, however often it is copied.
All rewards in a batch use its pre-batch counts; completed outcomes update
counts only after that batch. Pooling counts across tasks or graph classes
uses a fixed pooling rule.

\begin{cflemma}[Reward support and the posterior's model]
\label{lem:cf_reward_support}
The counts in Eq.~\ref{eq:cf_reward_snapshot} give
$0<\widehat p_k(G)<1$ and
$\varepsilon_{\min}\le\widetilde R_k\le\varepsilon_{\min}+1$.
Under a model of conditionally independent Bernoulli outcomes with one
fixed success parameter $p$ and prior $\operatorname{Beta}(1/2,1/2)$, the
posterior is $\operatorname{Beta}(s+1/2,n-s+1/2)$ and its mean is
$(s+1/2)/(n+1)$.
\end{cflemma}
\begin{proof}
The numerator is strictly between zero and $n+1$, proving the first claim.
The remaining multiplicative factors belong to $[0,1]$, proving the reward
bounds. For the posterior, multiply the prior density proportional to
$p^{-1/2}(1-p)^{-1/2}$ by the Bernoulli likelihood
$p^s(1-p)^{n-s}$ and normalize. Its two beta parameters and their ratio
give the stated distribution and mean.
\end{proof}
\noindent The positive floor keeps log rewards finite and gives failed
outcomes a small positive target mass.

\begin{cflemma}[Counter-induced target drift]
\label{lem:cf_counter_drift}
After $a\ge0$ new unique outcomes with $b\in\{0,\ldots,a\}$ successes,
\begin{equation}
 \widehat p_{n+a}-\widehat p_n
 =\frac{b-a\widehat p_n}{n+a+1},\qquad
 |\widehat p_{n+a}-\widehat p_n|\le\frac{a}{n+a+1}.
 \label{eq:cf_counter_drift}
\end{equation}
If the answer-score table and structural code are fixed between two reward
snapshots, write $r(G)=\varepsilon_{\min}+c_G\widehat p(G)$ with
$c_G=r_{\rm ans}(G)e^{-\lambda L_{\rm rel}(G)}\in[0,1]$. The per-graph
log-reward difference between the two snapshots, with rewards $r$ and $r'$, is then at most
\begin{equation}
 \frac{c_G\,a}{(n+a+1)\min\{r,r'\}}\le\frac{a}{(n+a+1)\varepsilon_{\min}}.
 \label{eq:cf_log_reward_drift}
\end{equation}
A uniform such bound can be substituted for $d_R$ in
Corollary~\ref{cor:cf_perturb}.
\end{cflemma}
\begin{proof}
Subtract the two ratios, writing $s+1/2=(n+1)\widehat p_n$, to obtain the
identity. Since $0\le b\le a$ and $0<\widehat p_n<1$, its numerator has
absolute value at most $a$. The reward difference is
$c_G(\widehat p'-\widehat p)$, so $|r'-r|\le c_Ga/(n+a+1)$. Between $r$ and $r'$
the derivative of $\log$ is at most $1/\min\{r,r'\}$, and the mean-value
theorem gives the first bound; $c_G\le1$ and $\min\{r,r'\}\ge\varepsilon_{\min}$
give the second.
\end{proof}

\begin{cfcorollary}[Movement of the self-generated target]
\label{cor:cf_target_drift}
Fix the answer-score table, structural code, graph support, and $\beta$
between snapshots $k$ and $k+1$, and let $q_k(G)\propto r_k(G)^\beta$ be the
target of snapshot $k$. If team $G$ gains $a_G$ unique outcomes between the
snapshots and
$d=\max_G c_Ga_G/\big((n_k(G)+a_G+1)\min\{r_k(G),r_{k+1}(G)\}\big)$, then
\begin{equation}
 e^{-2\beta d}\le\frac{q_{k+1}(G)}{q_k(G)}\le e^{2\beta d},\qquad
 \cfTV(q_{k+1},q_k)\le\tanh(\beta d/2).
 \label{eq:cf_target_drift}
\end{equation}
\end{cfcorollary}
\begin{proof}
Lemma~\ref{lem:cf_counter_drift} gives
$|\log r_{k+1}(G)-\log r_k(G)|\le d$ for every $G$. Hence
$q_{k+1}(G)=q_k(G)\tau_G/\cfE_{q_k}\tau$ with
$\tau_G=(r_{k+1}(G)/r_k(G))^\beta\in[e^{-\beta d},e^{\beta d}]$.
This is the tilt of Theorem~\ref{thm:cf_approx} with $\eta=\beta d$, and its
ratio and total-variation bounds apply unchanged.
\end{proof}
\noindent Teams with $a_G=0$ or $c_G=0$ contribute nothing to $d$, so $d$ is set
by the sampled teams with positive answer scores and shrinks as their records
accumulate. It holds for any realized counts, whatever policy selected the
teams.

\subsection{The practical training loss}
\label{app:cf_training_limits}

For a fixed replay distribution $\nu_k$ and fixed state weighting $d_k$, the
regularized training surrogate is
\begin{equation}
 \mathcal L_k(\vartheta)=
 \cfE_{\zeta\sim\nu_k}\frac{\rho_\vartheta(\zeta)^2}{T(\zeta)^2}
 +\alpha_{\rm KL}\cfE_{v\sim d_k}
 \cfKL(\bar\pi_\theta(\cdot\mid v)\|\pi_{\rm ref}(\cdot\mid v)).
 \label{eq:cf_regularized_loss}
\end{equation}
Here $\vartheta=(\theta,\phi,\psi)$ collects all learned parameters, the
reference policy is normalized on the same action domain and positive
where the forward policy is positive, and $\alpha_{\rm KL}\ge0$.
During an update, replay samples and recorded reward targets are treated as
fixed. Shared parameters must be differentiated in \emph{all} terms where
they occur. This specification distinguishes a replay gradient from the
full derivative of an on-policy sampling distribution.

\begin{cfexample}[KL regularization tilts the balance point]
\label{ex:cf_kl}
Consider two one-step terminals with unit rewards and backward probabilities
one. Let the forward law be $(p,1-p)$ and the reference $(0.9,0.1)$.
At $p=1/2$, $Z=2$, both TB residuals vanish and the squared-residual
gradient is zero, while
\[
 \left.\frac{d}{dp}\alpha_{\rm KL}\cfKL((p,1-p)\|(0.9,0.1))
 \right|_{p=1/2}=-\alpha_{\rm KL}\log9.
\]
With $\alpha_{\rm KL}>0$ the regularized optimum therefore moves from the
reward-proportional point toward a fixed reference, and $\alpha_{\rm KL}$ sets
the size of this move.
\end{cfexample}

\begin{cflemma}[Proximal reference keeps balanced directors fixed]
\label{lem:cf_proximal}
Let $\pi_{\rm ref}=\bar\pi_{\theta_k}$ be the director that collected batch $k$,
and fix the rewards. If $\vartheta_k=(\theta_k,\phi_k,\psi_k)$ has zero residual
on every replayed path, then $\vartheta_k$ minimizes $\mathcal L_k$. Hence a
balanced director is a fixed point of the regularized update, and the KL term
only damps moves away from it.
\end{cflemma}
\begin{proof}
Both terms of Eq.~\ref{eq:cf_regularized_loss} are nonnegative. At
$\vartheta_k$ the residual term vanishes by assumption, and the KL term vanishes
because $\bar\pi_{\theta_k}=\pi_{\rm ref}$. Hence
$\mathcal L_k(\vartheta_k)=0\le\mathcal L_k(\vartheta)$ for every $\vartheta$.
\end{proof}
\noindent Because the reference moves with the director, the tilt of the
preceding example vanishes at every balanced director, and the recursion keeps
the reward-proportional point of each snapshot. This choice mirrors the
proximal updates of TRPO and PPO \citep{schulman2015trust,schulman2017proximal},
which regularize each policy update toward the policy that collected the data.

\begin{cftheorem}[Random log-reward regression]
\label{thm:cf_random_reward}
Fix a replay path $\zeta$ ending at $G$, its length, and its context.
Assume execution randomness is fresh conditional on $G,c$ and independent
of the construction path, and $\log\widetilde R_k(G,x,\xi)$ has finite second
moment. Put $m_G=\cfE[\log\widetilde R_k\mid G,c]$,
$\sigma_G^2=\cfVar(\log\widetilde R_k\mid G,c)$, and
$A_\vartheta(\zeta)=\log Z+\log P_F(\zeta)-\log Q_G(\zeta)$. Then
\begin{equation}
 \cfE\!\left[\frac{(A_\vartheta(\zeta)-\beta\log\widetilde R_k)^2}{T^2}
 \,\middle|\,\zeta,c\right]
 =\frac{(A_\vartheta(\zeta)-\beta m_G)^2}{T^2}
 +\frac{\beta^2\sigma_G^2}{T^2}.
 \label{eq:cf_noise_decomposition}
\end{equation}
In an unregularized fixed-replay model that can simultaneously realize all
these conditional means, the corresponding ideal graph weights are
$\exp(\beta m_G)$, the tempered geometric-mean reward of each team.
\end{cftheorem}
\begin{proof}
Write $\log\widetilde R_k=m_G+\epsilon$ with conditional mean-zero
$\epsilon$ and variance $\sigma_G^2$. Expanding the square makes the cross
term vanish and gives the identity. The variance term is independent of
$A_\vartheta$ for fixed replay and context. If all means can be realized,
the remaining squared terms are minimized at $A_\vartheta=\beta m_G$.
These are exactly the fixed-weight path identities with
$w_G=\exp(\beta m_G)$, to which the appropriate normalized or valid-only
TB theorem applies.
\end{proof}
\noindent\emph{Example (CTB favors reliable teams).} One team with reward
$0.1$ or $0.9$ equally often and another with constant reward $0.5$ have the
same mean reward, but their geometric-mean weights are $0.3$ and $0.5$. For
$\beta=1$, CTB samples them with probabilities $0.375$ and $0.625$ in the
realizable two-terminal model, so under noisy rewards it prefers the team
that succeeds consistently.

\begin{cflemma}[Residual-gradient bound]
\label{lem:cf_gradient}
For a fixed sampled path, fixed reward, and differentiable finite residual,
\begin{equation}
 g(\zeta):=\nabla_\vartheta L_{\rm CTB}(\zeta)
 =\frac{2\rho(\zeta)}{T(\zeta)^2}\nabla_\vartheta\rho(\zeta).
 \label{eq:cf_residual_gradient}
\end{equation}
If $T\ge T_{\min}\ge1$ and
$\|\nabla_\vartheta\rho(\zeta)\|\le G_0<\infty$ on the distribution under
consideration, then
\begin{equation}
 \cfE\|g-\cfE g\|^2
 \le\cfE\|g\|^2
 \le\frac{4G_0^2}{T_{\min}^4}\cfE[\rho^2].
 \label{eq:cf_gradient_bound}
\end{equation}
\end{cflemma}
\begin{proof}
Treating the recorded length as fixed, the chain rule gives the gradient.
Squaring its norm and using the two uniform bounds yields the last
inequality after expectation. Subtracting $\|\cfE g\|^2\ge0$ proves the
first. Thus, as the mean-square residual shrinks, the variance of this CTB
gradient component shrinks with it.
\end{proof}

\subsection{Occupancy flows, conditional bridges, and path diagnostics}
\label{app:cf_occupancy}

Fix the normalized outcome model of Section~\ref{app:terminal_distribution}.
For any learned forward policy, the reverse kernels define the
reward-weighted path measure
\begin{equation}
 M(\zeta)=w_zQ_z(\zeta),\qquad
 \sum_\zeta M(\zeta)=W,\qquad P^\star(\zeta)=M(\zeta)/W.
 \label{eq:cf_path_measure}
\end{equation}
This measure allocates the target over paths. For a construction vertex $v$ or transition $e$, define
\begin{equation}
 \begin{aligned}
 F_z(v)&=\sum_{\zeta\in\Omega_z:\,v\in\zeta}M(\zeta),&
 f_z(e)&=\sum_{\zeta\in\Omega_z:\,e\in\zeta}M(\zeta),\\
 F(v)&=\sum_zF_z(v),& f(e)&=\sum_zf_z(e).
 \end{aligned}
 \label{eq:cf_occupancy_definition}
\end{equation}
States and transitions occur at most once on a construction-DAG path.
Source flow is outgoing flow $F(v_0)=W$; terminal flow is incoming flow
$F(g_z)=w_z$. Incoming and outgoing flows agree at every interior
vertex.

\begin{cftheorem}[Occupancy conservation and terminal-conditioned bridge]
\label{thm:cf_occupancy}
Under T1--T3, the path measure in Eq.~\ref{eq:cf_path_measure} satisfies
$F_z(v)=\sum_{e\in\cfIn(v)}f_z(e)=\sum_{e\in\cfOut(v)}f_z(e)$ at every
interior $v$, and
$f_z(e)=F_z(v')b_z(e\mid v')$ for $e=(v,u,v')$ in the
$z$-ancestral subgraph, with zero occupancy elsewhere.
If additionally exact balance holds, then
\[
 F(v)=W P_F(v\text{ is visited}),\qquad
 f(e)=F(v)p_F(e\mid v).
\]
For a nonterminal visited state $v$ with
$h_z(v):=P_F(z\mid v\text{ is visited})>0$, define the forward bridge
\begin{equation}
 p_F^z(e\mid v)=p_F(e\mid v)\frac{h_z(v')}{h_z(v)}.
 \label{eq:cf_bridge}
\end{equation}
It is normalized over outgoing edges and obeys
\begin{equation}
 f_z(e)=F_z(v)p_F^z(e\mid v)=F_z(v')b_z(e\mid v').
 \label{eq:cf_conditional_balance}
\end{equation}
Thus the bridge supplies the forward side of an endpoint-conditioned
occupancy.
\end{cftheorem}
\begin{proof}
For each interior visit of a complete path, exactly one incoming and one
outgoing edge are used. Summing these pathwise indicator identities against
$M$ proves conservation. In the reverse process starting at $g_z$, once
$v'$ is visited the next reverse edge has law $b_z(\cdot\mid v')$.
All remaining reverse prefixes have mass one by
Lemma~\ref{lem:cf_reverse_mass}; hence its edge-visit mass is
$F_z(v')b_z(e\mid v')$.

At zero residual, $M(\zeta)=WP_F(\zeta)$. Summing gives the state-visit
identity. A sufficient Markov state makes the next-edge law independent of
the earlier prefix, so $f(e)=F(v)p_F(e\mid v)$. The same Markov property
gives $h_z(v)=\sum_{e=(v,u,v')}p_F(e\mid v)h_z(v')$, proving bridge
normalization. Furthermore $F_z(v)=F(v)h_z(v)$ and
$f_z(e)=F(v)p_F(e\mid v)h_z(v')$. Substitution proves the two-sided
conditional balance.
\end{proof}
\noindent\emph{Example.} For a source that chooses between two unit-reward
terminals with probability $1/2$ each, the bridge toward the first terminal
takes that branch with probability one, so
$F_1(v_0)p_F^1(g_1\mid v_0)=1=F_1(g_1)b_1$.

\paragraph{State-flow ratios and edge flows.}
Whenever a compatible local balance identity holds, its ratio is a ratio
of state occupancies, $p_F^z/b_z=F_z(v')/F_z(v)$ on positive conditional
support, while the reward an edge carries is its edge flow $f_z(e)$. For two
parents of flows $1$ and $99$ entering a common child of flow $100$
deterministically, the backward probabilities are $0.01$ and $0.99$, the
flow ratios are $100$ and $100/99$, and the edge flows are $1$ and $99$; we
therefore rank edges by $f_z(e)$.

\begin{cfdefinition}[Semantic transition occupancy]
For a semantic communication decision $a=(i,j,p)$, let
$N_a(\zeta)$ count accepted insertions of that edge on a valid complete
trajectory and set $N_a(\zeta)=0$ on abort trajectories. Define
\begin{equation}
 H_c(a)=\sum_\zeta M(\zeta)N_a(\zeta)
       =W\cfE_{P^\star}N_a(\zeta).
 \label{eq:cf_semantic_occupancy}
\end{equation}
Under the add-only interface with duplicate insertions excluded,
$N_a(\zeta)=\cfone\{a\in E(G_\zeta)\}$, so
$H_c(a)=\sum_{G:\,a\in E(G)}w_G$.
The normalized statistic $H_c(a)/W$ is the target probability of returning
a valid team containing that edge; $H_c(a)/W_{\mathsf V}$ conditions this
probability on valid completion.
\end{cfdefinition}
The last identity follows by summing $Q_G$ over $\Omega_G$. It is independent
of the allocation between different orders for a fixed graph, so semantic
occupancy reads the canonical team rather than any one construction order.

\paragraph{Estimators and their sampling law.}
If $\zeta_1,\ldots,\zeta_n$ are independent samples from $P^\star$, then
$n^{-1}\sum_iN_a(\zeta_i)$ is unbiased for $H_c(a)/W$ by linearity of
expectation. Off policy, for a normalized proposal $\nu$ positive wherever
$M(\zeta)N_a(\zeta)>0$, the estimator
\begin{equation}
 \widehat H_c(a)=\frac1n\sum_{i=1}^n
 \frac{M(\zeta_i)}{\nu(\zeta_i)}N_a(\zeta_i),\qquad\zeta_i\sim\nu,
 \label{eq:cf_occupancy_estimator}
\end{equation}
is unbiased for $H_c(a)$, since summing $\nu(\zeta)$ cancels the denominator.
In raw valid-only exact balance,
$W_{\mathsf V}$ and the valid-conditioned target replace $W$ and $P^\star$.
At finite residual, unweighted forward frequencies describe the learned
policy.

\paragraph{Cached path diagnostics and computation.}
For a scored path define
$d_t(\zeta)=\log p_F(e_t\mid v_{t-1})-\log b_z(e_t\mid v_t)$.
Its complete sum has the exact algebraic identity
\begin{equation}
 \sum_{t=1}^T d_t(\zeta)=\log w_z-\log Z+\rho(\zeta).
 \label{eq:cf_path_diagnostic}
\end{equation}
This path diagnostic holds at any residual. The $d_t$ values are computed by
arithmetic from cached forward and backward log probabilities.

\subsection{Proof dependencies}
\label{app:cf_closure}

The preceding results form one dependency chain. Checked construction
determines the legal support. Normalized action and reverse kernels turn this
support into path laws. A fixed terminal identity and reward snapshot make the
TB path sum meaningful, and exact balance on all paths (T4) or on valid paths
(T4$'$), or the uniform residual bounds, link these laws to the prescribed
terminal distribution. The
reward-snapshot corollary bounds how this target moves between rounds, and the
communication and single-writer results are runtime contracts that hold
alongside this chain.

\begin{center}
\small
\renewcommand{\arraystretch}{1.12}
\begin{tabular}{>{\raggedright\arraybackslash}p{0.21\linewidth}>{\raggedright\arraybackslash}p{0.34\linewidth}>{\raggedright\arraybackslash}p{0.36\linewidth}}
\toprule
Result & Premises & Main-text result \\
\midrule
Construction & C1 for soundness; C1--C3 for coverage & Proposition~\ref{prop:construction}: valid canonical teams, every target reachable \\
Evidence update & Version-bound, authentic applicable records; ordered gate & Proposition~\ref{prop:revision}: version-bound evidence, flip and correction identities \\
Runtime & Fixed schedule, finite rounds, charged bounded calls & Finite, budgeted execution \\
Single writer & S1--S3 & Proposition~\ref{prop:revision}: serialized internal writes \\
Graph sampling & T1--T3 with T4$'$ (Algorithm~\ref{alg:app_training}) or T4; uniform residual control & Proposition~\ref{prop:terminal_distribution}: reward-proportional teams \\
Proximal update & $\pi_{\rm ref}$ is the director that collected the batch & Balanced directors stay fixed (Section~\ref{subsec:rl}) \\
Target movement & Fixed answer-score table, structural code, and team support & Proposition~\ref{prop:bounded_target}: bounded target drift \\
Occupancy & Normalized reward path measure & Path diagnostics \\
\bottomrule
\end{tabular}
\end{center}

All finite sums above use the declared supports, all reverse normalizers
include parent--action multiplicity, and all probability comparisons hold
at a fixed context. Conditional events of probability zero use stated
conventions.